\documentclass[11pt]{article}

\usepackage[hyphens]{url}
\usepackage[colorlinks,pagebackref]{hyperref}
\usepackage[hyphenbreaks]{breakurl}
\usepackage{amsmath} 
\usepackage{amsthm} 
\usepackage{amssymb}	
\usepackage{graphicx} 
\usepackage{multicol} 
\usepackage{multirow}
\usepackage{color}
\usepackage[dvips,letterpaper,margin=1in,bottom=1in]{geometry}
\usepackage[capitalize,noabbrev]{cleveref}

\usepackage{ytableau}

\usepackage{thmtools}

\usepackage[utf8]{inputenc}
\usepackage[english]{babel}

\usepackage[T1]{fontenc}
\AtBeginDocument{%
  \DeclareFontShape{T1}{cmr}{m}{scit}{<->ssub*cmr/m/sc}{}%
}

\usepackage{diagbox}
\usepackage{mathtools}

\newtheorem{theorem}{Theorem}[section]

\newtheorem{lemma}[theorem]{Lemma}
\newtheorem{corollary}[theorem]{Corollary}

\newtheorem{fact}[theorem]{Fact}

\newtheorem{definition}[theorem]{Definition}

\newtheorem{remark}{Remark}[section]
\newtheorem{problem}{Problem}

\DeclarePairedDelimiter\rbra{\lparen}{\rparen}
\DeclarePairedDelimiter\sbra{\lbrack}{\rbrack}
\DeclarePairedDelimiter\cbra{\{}{\}}
\DeclarePairedDelimiter\abs{\lvert}{\rvert}
\DeclarePairedDelimiter\Abs{\lVert}{\rVert}
\DeclarePairedDelimiter\ceil{\lceil}{\rceil}
\DeclarePairedDelimiter\floor{\lfloor}{\rfloor}
\DeclarePairedDelimiter\ket{\lvert}{\rangle}
\DeclarePairedDelimiter\bra{\langle}{\rvert}

\DeclareMathOperator*{\E}{\mathbb{E}}

\newcommand{\tr} {\operatorname{tr}}
\newcommand{\poly} {\operatorname{poly}}
\newcommand{\diag} {\operatorname{diag}}
\newcommand{\polylog} {\operatorname{polylog}}

\newcommand{\spanspace} {\operatorname{span}}

\makeatletter
\def\@buildmath#1{%
  \expandafter\def\csname bb#1\endcsname{\ensuremath{\mathbb{#1}}}%
  \expandafter\def\csname bf#1\endcsname{\ensuremath{\mathbf{#1}}}%
  \expandafter\def\csname sf#1\endcsname{\ensuremath{\mathsf{#1}}}%
  \expandafter\def\csname cal#1\endcsname{\ensuremath{\mathcal{#1}}}%
  \expandafter\def\csname rm#1\endcsname{\ensuremath{\mathrm{#1}}}%
  \expandafter\def\csname tt#1\endcsname{\ensuremath{\mathtt{#1}}}%
}
\def\@buildmathletters#1{%
  \ifx#1\relax\else
    \@buildmath{#1}%
    \expandafter\@buildmathletters
  \fi
} 
\@buildmathletters ABCDEFGHIJKLMNOPQRSTUVWXYZabcdefghijklmnopqrstuvwxyz\relax
\makeatother

\usepackage{optidef}

\usepackage{xspace}
\newcommand{\BQP}{\textnormal{\textsf{BQP}}\xspace}

\newcommand{\NIQSZK}{\textnormal{\textsf{NIQSZK}}\xspace}

\newcommand{\PS}{\mathrm{PS}}

\newcommand{\td}{\mathrm{T}}
\newcommand{\TV}{\mathrm{TV}}
\newcommand{\Hq}{\mathrm{H}^\mathrm{T}_q}
\newcommand{\HqHalf}{\mathrm{H}^\mathrm{T}_{1/2}}
\newcommand{\Sq}{\mathrm{S}^\mathrm{T}_q}
\newcommand{\SqHalf}{\mathrm{S}^\mathrm{T}_{1/2}}

\newcommand{\QSCMM}{\textnormal{\textsc{QSCMM}}\xspace}
\newcommand{\TsallisQEA}{\texorpdfstring{\textnormal{\textsc{TsallisQEA}\textsubscript{\textit{q}}}}\xspace}
\newcommand{\TsallisQEAnoq}{\texorpdfstring{\textnormal{\textsc{TsallisQEA}}}\xspace}

\newcommand{\ketbra}[2]{\ensuremath{\ket{#1}\!\bra{#2}}}

\usepackage{latexsym}
\usepackage{CJK}

\usepackage{enumerate}

\usepackage{algorithm}
\usepackage{algpseudocode}

\usepackage{stmaryrd}
\usepackage{tabularx}
\usepackage{booktabs}
\usepackage{adjustbox}

\newcommand{\footremember}[2]{%
    \footnote{#2}
    \newcounter{#1}
    \setcounter{#1}{\value{footnote}}%
}

\allowdisplaybreaks

\usepackage{tikz}
\usetikzlibrary{quantikz2}

\begin{document}

\title{Trace Estimation of Quantum State Powers: \\ {Sample Complexity and Computational Hardness}\footnotetext{A preliminary version of this paper \cite{CW25} was presented at the 38th Conference on Learning Theory (COLT 2025).}}
\author{Kean Chen \footremember{1}{\url{keanchen.gan@gmail.com}}
    \and Yupan Liu \footremember{2}{\url{yupan.liu@epfl.ch}}
    \and Qisheng Wang \footremember{3}{\url{QishengWang1994@gmail.com}}}
\date{}

\maketitle

\begin{abstract}
    As often emerges in various basic quantum properties such as R\'enyi and Tsallis entropies, 
    the trace of quantum state powers $\tr\rbra{\rho^q}$ has attracted a lot of attention.
    The recent work of \hyperlink{cite.LW25}{Liu and Wang} \hyperlink{cite.LW25}{(SODA 2025)} showed that, even for (possibly) non-integer $q>1$, $\tr\rbra{\rho^q}$ can be estimated to within additive error $\varepsilon$ using a dimension-independent (and also rank-independent) sample complexity of $\widetilde O(1/\varepsilon^{3+\frac{2}{q-1}})$,\footnote{Throughout this paper, $\widetilde O\rbra{\cdot}$, $\widetilde \Omega\rbra{\cdot}$, and $\widetilde \Theta\rbra{\cdot}$ suppress polylogarithmic factors in $\varepsilon$.} together with a lower bound of $\Omega(1/\varepsilon)$. In addition, combining this result with subsequent work of \hyperlink{cite.Liu26}{Liu (STACS 2026)} shows that the corresponding promise problem is \textsf{BQP}-complete.
    In this paper, we significantly improve and extend the sample complexity bounds for this problem. Furthermore, we show that for $0<q<1$, the problem does not admit an efficient estimator unless $\mathsf{BQP}=\mathsf{NIQSZK}$, which is considered highly unlikely.
    In particular, we have the following results.
    \begin{itemize}
        \item For $q > 2$, we settle the sample complexity with matching upper and lower bounds $\widetilde \Theta\rbra{1/\varepsilon^2}$. 
        \item For $1 < q < 2$, we obtain an upper bound of $\widetilde O\rbra{1/\varepsilon^{\frac{2}{q-1}}}$, with a lower bound of $\Omega\rbra{1/\varepsilon^{\max\cbra{\frac{1}{q-1}, 2}}}$ for dimension-independent (in fact, rank-independent) estimators.
        \item For $0< q < 1$, we obtain an upper bound of $O((d/\varepsilon)^{\frac{2}{q}})$, with a lower bound of $\Omega((d/\varepsilon)^{\frac{1}{q}})$ for $d$-dimensional states (in fact, both bounds can be naturally refined to depend on the rank rather than the dimension). Accordingly, the corresponding promise problem is \textsf{NIQSZK}-hard, which is in sharp contrast to the case of $q > 1$. 
    \end{itemize}
    Technically, our upper bounds are obtained by (non-plug-in) quantum estimators based on weak Schur sampling, in sharp contrast to the prior approach based on quantum singular value transformation and samplizer. 
\end{abstract}

\newpage
\tableofcontents
\newpage

\section{Introduction}

Testing the properties of quantum states is a fundamental problem in the field of quantum property testing \cite{MdW16}, where the spectra of quantum states turn out to be crucial, as they fully characterize unitarily invariant properties. 
Given samples of the quantum state to be tested, in \cite{OW21}, testing the spectrum was extensively studied, with several significant applications such as mixedness testing and rank testing.
In \cite{OW17}, they further investigated the sample complexity of the spectrum tomography of quantum states. 
Subsequently, as a representative unitarily invariant quantity, the entropy of a quantum state was known to have efficient estimators in \cite{AISW20,BMW16,WZ24}. 

The traces of quantum state powers, $\tr\rbra{\rho^q}$, of a quantum state $\rho$ are one of the simplest functionals of quantum states. 
The quantity $\tr\rbra{\rho^q}$ has  connections to the R\'enyi entropy $\mathrm{S}_q^{\mathrm{R}}\rbra{\rho} = \frac{1}{1-q} \ln\rbra{\tr\rbra{\rho^q}}$ \cite{Ren61} and the Tsallis entropy $\mathrm{S}_q^{\mathrm{T}}\rbra{\rho} = \frac{1}{1-q}\rbra{\tr\rbra{\rho^q} - 1}$ \cite{Tsa88}. 
The estimation of $\tr\rbra{\rho^q}$ is at the core of Tsallis entropy estimation, with a wide range of applications in physics.
A notable example is the Tsallis entropy of order $q = \frac{3}{2}$ for modeling fluid dynamics systems \cite{Bec01,Bec02}. 
In addition, for $q = 1.001$ (close to $1$), the Tsallis entropy $\mathrm{S}_q^{\mathrm{T}}\rbra{\rho}$ serves as a lower bound on the von Neumann entropy, whereas the former can be estimated exponentially faster than the latter, as noted in \cite{LW25}. 
In particular, $\tr\rbra{\rho^2}$ refers to the purity of $\rho$,
and it is well-known that the purity $\tr\rbra{\rho^2}$ can be estimated to within additive error $\varepsilon$ using $O\rbra{1/\varepsilon^2}$ samples of $\rho$ via the SWAP test \cite{BCWdW01}. 
For the case of constant integer $q \geq 2$, $\tr\rbra{\rho^q}$ can be estimated using $O\rbra{1/\varepsilon^2}$ samples of $\rho$ via the Shift test proposed in \cite{EAO+02}, generalizing the SWAP test. 
For non-integer $q > 0$ and $q \neq 1$, the estimation of $\tr\rbra{\rho^q}$ was considered in \cite{WGL+24} with the corresponding quantum algorithms presented with time complexity $\poly\rbra{r, 1/\varepsilon}$,\footnote{In \cite{WGL+24}, their main results only consider the quantum query complexity, as they assume access to the state-preparation circuit of $\rho$. Nevertheless, their results also imply a sample complexity of $\poly\rbra{r, 1/\varepsilon}$ (with a polynomial overhead compared to the corresponding query complexity) using the techniques in \cite{GP22}, as noted in \cite[Footnote 2]{WGL+24}.} where $r$ is the rank of $\rho$. 
Recently in \cite{LW25}, it was discovered that for every non-integer $q > 1$, $\tr\rbra{\rho^q}$ can be estimated using $\widetilde O\rbra{1/\varepsilon^{3+\frac{2}{q-1}}}$ samples of $\rho$, removing the dependence on $r$ (which we call dimension-independent as it depends on neither the rank nor the dimension of $\rho$). Thus, this exponentially improves the results in \cite{WGL+24} and the results implied by other works \cite{AISW20,WZL24,WZ24} on R\'enyi entropy estimation.
However, the sample complexity in \cite{LW25} is far from being optimal, as only a lower bound of $\Omega\rbra{1/\varepsilon}$ 
on the sample complexity of estimating $\tr\rbra{\rho^q}$ for non-integer $q > 1$ 
was known in \cite[Theorem 5.9]{LW25}. 
To our knowledge, only a matching lower bound of $\Omega\rbra{1/\varepsilon^2}$ was known for the case of $q = 2$, i.e., estimating the purity $\tr\rbra{\rho^2}$ (see \cite[Theorem 5]{CWLY23} and \cite[Lemma 3]{GHYZ24}).

Beyond prior work on the query and sample complexities of estimating the von Neumann entropy and the Tsallis entropy $\Sq(\rho)$, it is natural to consider the corresponding promise problem that distinguishes whether $\Sq(\rho) \geq \tau_0$ or $\Sq(\rho) \leq \tau_1$ for constants $\tau_0>\tau_1$. In the von Neumann entropy case ($q=1$), this problem is \NIQSZK{}-complete~\cite{Kobayashi03,BASTS10,CCKV08}, a class that originally arose in the study of (non-interactive) quantum statistical zero-knowledge~\cite{Watrous02,Kobayashi03}. Regardless of its origin, this result implies that estimating the von Neumann entropy is computationally equivalent to distinguishing any efficiently preparable quantum state from the maximally mixed state with respect to the trace distance.\footnote{While our work focuses on the case where $q$ is a real number, it is noteworthy that the \NIQSZK{}-hardness result extends to the regime where $q$ is slightly above $1$, specifically $1 \leq q \leq 1+\frac{1}{n}$, as proven in~\cite{LW25}.} 
In contrast, for constant $q>1$, the complexity landscape changes in a manner consistent with the known efficient estimators: it was shown in~\cite{LW25} that the regime $q\in(1,2]$ is \BQP{}-complete, while the regime $q>2$ lies in \BQP{}. An open problem posed in the conference version of this work~\cite{CW25}, asking whether constant $q>2$ is also \BQP{}-hard, was resolved affirmatively in~\cite{Liu26}. Together with~\cite{LW25}, this establishes that the associated problem is \BQP{}-complete for all constant $q>1$. 

In this paper, we further investigate the sample complexity and the computational hardness of estimating $\tr\rbra{\rho^q}$.
For $q > 1$, we achieve significant improvements over the prior results \cite{LW25} in both the upper and lower bounds.
In particular, for $q > 2$, we provide an estimator that is \textit{optimal} only up to a logarithmic factor in the precision $\varepsilon$.
For $0 < q < 1$, we complement the literature by providing estimators and showing that the corresponding promise problem is $\NIQSZK{}$-hard, which is in sharp contrast to the $\BQP{}$-completeness for $q > 1$ due to \cite{LW25,Liu26}. 
Our results are collected in \cref{sec:result}.
In addition, it is noteworthy that our techniques are conceptually and technically different from those in \cite{LW25}. 
In comparison, our estimator is based on weak Schur sampling \cite{CHW07} while the estimator in \cite{LW25} is based on quantum singular value transformation \cite{GSLW19} and samplizer \cite{WZ23,WZ24}. 
Concerning the computational hardness, while our result follows the framework of~\cite{CCKV08,KLGN19} at a conceptual level, our main technical contribution is a new inequality bounding the $q$-Tsallis entropy for $0<q<1$ in terms of its closeness to the uniform distribution, which is of independent interest.
For more details, see \cref{sec:technique}. 

\subsection{Main Results} \label{sec:result}

To illustrate our results, we present them in three parts separately: $q > 2$, $1 < q < 2$, and $0 < q < 1$.

\paragraph{The case of $q > 2$.}
For $q > 2$, we provide a quantum estimator with optimal sample complexity $\widetilde \Theta\rbra{1/\varepsilon^2}$ only up to a logarithmic factor in $\varepsilon$. 

\begin{theorem} [Optimal estimator for $q > 2$, informal version of \cref{thm:ub-q>2,thm:lb}] \label{thm:q>2-main}
    For every $q > 2$, it is necessary and sufficient to use $\widetilde\Theta\rbra{1/\varepsilon^2}$ samples of the quantum state $\rho$ to estimate $\tr\rbra{\rho^q}$ to within additive error $\varepsilon$.
\end{theorem}

\paragraph{The case of $1 < q < 2$.}
For $1 < q < 2$, we provide a quantum estimator with sample complexity $\widetilde O\rbra{1/\varepsilon^{\frac{2}{q-1}}}$, only with room for quadratic improvements due to a lower bound of $\Omega\rbra{1/\varepsilon^{\max\cbra{\frac{1}{q-1},2}}}$.

\begin{theorem} [Improved estimator for $1 < q < 2$, informal version of \cref{thm:ub-1<q<2,thm:lb,thm:lb1to2}] \label{thm:q<2-main}
    For every $1 < q < 2$, it is sufficient to use $\widetilde O\rbra{1/\varepsilon^{\frac{2}{q-1}}}$ samples of the quantum state $\rho$ to estimate $\tr\rbra{\rho^q}$ to within additive error $\varepsilon$.
    On the other hand, when the dimension of $\rho$ is sufficiently large, $\Omega\rbra{1/\varepsilon^{\max\cbra{\frac{1}{q-1},2}}}$ samples of $\rho$ are necessary. 
\end{theorem}

Both \cref{thm:q>2-main,thm:q<2-main} improve the prior best upper bound $\widetilde O\rbra{1/\varepsilon^{3+\frac{2}{q-1}}}$ and lower bound $\Omega\rbra{1/\varepsilon}$ in \cite{LW25}.  
It is also noted that \cref{thm:q>2-main} gives a matching lower bound of $\Omega\rbra{1/\varepsilon^2}$ on the sample complexity of estimating $\tr\rbra{\rho^q}$ for every integer $q \geq 3$, implying that the Shift test in \cite{EAO+02} is sample-optimal to estimate $\tr\rbra{\rho^q}$ to within an additive error, generalizing the lower bounds in \cite{CWLY23,GHYZ24} for the optimality of the SWAP test \cite{BCWdW01} to estimate $\tr\rbra{\rho^2}$. 

\paragraph{The case of $0 < q < 1$.}
For $0 < q < 1$, we provide a quantum estimator with sample complexity $O\rbra{d^{2/q}/\varepsilon^{2/q}}$, only with room for quadratic improvements due to a lower bound of $\Omega\rbra{d^{1/q}/\varepsilon^{1/q}}$. 

\begin{theorem}[Estimator for $0 < q < 1$, informal version of \cref{thm-2100931,thm-03031025}] \label{thm:q<1-main}
    For every $0 < q < 1$, it is sufficient to use $O\rbra{d^{2/q}/\varepsilon^{2/q}}$ samples of the $d$-dimensional quantum state $\rho$ to estimate $\tr\rbra{\rho^q}$ to within additive error $\varepsilon$.
    On the other hand,  $\Omega\rbra{d^{1/q}/\varepsilon^{1/q}}$ samples of $\rho$ are necessary.
\end{theorem}

In \cref{thm:q<1-main}, if the rank of $\rho$ is known to be at most $r$ in advance, then we can replace $d$ in our upper and lower bounds with $r$.
It is noteworthy that \cref{thm:q<1-main} also implies an estimator for the R\'enyi entropy $\mathrm{S}_q^{\mathrm{R}}\rbra{\rho}$ for $0 < q < 1$ with sample complexity $O\rbra{d^{2/q}/\varepsilon^{2/q}}$,\footnote{Any estimator for the Tsallis entropy $\mathrm{S}_q^{\mathrm{T}}\rbra{\rho}$ is an estimator for the R\'enyi entropy $\mathrm{S}_q^{\mathrm{R}}\rbra{\rho}$ to the same additive error for $0 < q < 1$.} thus giving an alternative proof of \cite[Theorem 4]{AISW20}. 

\begin{remark}[Time efficiency of the estimators in \cref{thm:q>2-main,thm:q<2-main,thm:q<2-main,thm:q<1-main}]
The estimators in \cref{thm:q>2-main,thm:q<2-main} can actually be implemented with quantum time complexity $\poly\rbra{\log\rbra{d}, 1/\varepsilon}$ for any constant $q > 1$. 
This is because, in their implementations (\cref{algo:q>2,algo:1<q<2}), we only need the first $m$ entries of the output of the quantum spectrum estimation (\cref{algo:spectrum}) with $m \leq O\rbra{1/\varepsilon^{\max\cbra{1, \frac{1}{q-1}}}}$. 
On the other hand, \cref{algo:spectrum} uses $\mathsf{n} = \widetilde{O}\rbra{1/\varepsilon^{\max\cbra{\frac{2}{q-1},2}}}$ samples of $\rho$ and can be implemented with quantum time complexity $O\rbra{\mathsf{n}^3\polylog\rbra{\mathsf{n}, d}} = \widetilde{O}\rbra{1/\varepsilon^{\max\cbra{\frac{6}{q-1},6}}} \cdot \polylog\rbra{d}$ by weak Schur sampling \cite{CHW07}.\footnote{This quantum time complexity was noted in \cite{WZ24c,WZ24,Hay24}. This is achieved by using the implementation of weak Schur sampling introduced in \cite[Section 4.2.2]{MdW16}, equipped with the quantum Fourier transform over symmetric groups in \cite{KS16}.}
Similarly, the estimator in \cref{thm:q<1-main} can be implemented in \cref{algo-2100311}, using $\mathsf{n} = O\rbra{\rbra{d/\varepsilon}^{2/q}}$  samples of $\rho$ and thus with quantum time complexity $O\rbra{\mathsf{n}^3 \polylog\rbra{\mathsf{n}, d}} = \widetilde{O}\rbra{\rbra{d/\varepsilon}^{6/q}} \cdot \polylog\rbra{d}$. 
\end{remark}

We summarize the developments for the sample complexity of estimating $\tr\rbra{\rho^q}$ in \cref{tab:summary}.

\renewcommand{\arraystretch}{1.5}
\begin{table}[!htp]
\centering
\caption{Sample complexity of estimating $\tr\rbra{\rho^q}$.}
\label{tab:summary}
\vspace{3pt}
\begin{tabular}{|ccc|c|}
\hline
\multicolumn{1}{|c|}{$q \geq 2$}               & $1 < q < 2$ & \multicolumn{1}{|c|}{$0 < q < 1$}                & References        \\ \hline
\multicolumn{1}{|c|}{$O(1/\varepsilon^2)$, $q \in \mathbb{N}$}  & /                   & \multicolumn{1}{|c|}{/} & \cite{BCWdW01,EAO+02}  \\ \hline
\multicolumn{1}{|c|}{$\Omega\rbra{1/\varepsilon^2}$, $q=2$} & / & \multicolumn{1}{|c|}{/}                  & \cite{CWLY23,GHYZ24}  \\ \hline
\multicolumn{3}{|c|}{$\poly(d,1/\varepsilon)$}                            & \cite{AISW20,WGL+24,WZL24,WZ24}  \\ \hline
\multicolumn{2}{|c|}{$\widetilde O\rbra{1/\varepsilon^{3+\frac{2}{q-1}}}$,~~$\Omega\rbra{1/\varepsilon}$}                                & / & \cite{LW25}   \\ \hline
\multicolumn{1}{|c|}{$\widetilde\Theta\rbra{1/\varepsilon^2}$}      & \multicolumn{1}{c|}{\begin{tabular}{c} $\widetilde O\rbra{1/\varepsilon^{\frac{2}{q-1}}}$ \\ $\Omega\rbra{1/\varepsilon^{\max\cbra{\frac{1}{q-1}, 2}}}$ \end{tabular}} & \begin{tabular}{c} $\widetilde O\rbra{d^{2/q}/\varepsilon^{2/q}}$ \\ $\Omega\rbra{d^{1/q}/\varepsilon^{1/q}}$ \end{tabular} & This Work     \\ \hline
\end{tabular}
\end{table}

Moreover, from a complexity-theoretic perspective, we establish the computational hardness of estimating $\tr\rbra{\rho^q}$ for $0 < q < 1$ as follows.

\begin{theorem}[Informal version of \Cref{thm:TsallisQEA-NIQSZKhard}]
    \label{thm:NIQSZK-informal}
    For every $q\in(0,1)$, the promise problem of estimating $\tr\rbra{\rho^q}$, \TsallisQEA{}, is \NIQSZK{}-hard. 
\end{theorem}

It follows that an efficient estimator cannot exist unless \NIQSZK{} collapses to \BQP{}, which is highly unlikely and thus consistent with the qualitative change in the sample complexity bounds across different regimes of $q$ in \Cref{tab:summary}; in particular, an efficient estimator exists only in those regimes of $q$ for which the corresponding promise problem is \BQP{}-complete. In summary, we collect the computational hardness results for \TsallisQEA{} in \cref{tab:hardness}.

\renewcommand{\arraystretch}{1}
\begin{table}[!htp]
    \centering
    \caption{Computational hardness of \TsallisQEA{}.}
    \label{tab:hardness}
    \begin{tabular}{cccc}
    \toprule
         $q > 2$ & $1 < q \leq 2$ & \begin{tabular}{c}
    $q = 1$ \\ (von Neumann) \end{tabular} & $0 < q < 1$ \\
    \midrule
    \begin{tabular}{c}
    \BQP{}-complete \\ \cite{LW25,Liu26}\end{tabular} & \begin{tabular}{c}
    \BQP{}-complete \\ \cite{LW25}\end{tabular} & \begin{tabular}{c}
    \NIQSZK{}-complete \\ \cite{BASTS10,CCKV08}\end{tabular} & \begin{tabular}{c}
    \NIQSZK{}-hard \\ This Work \end{tabular}  \\
    \bottomrule
    \end{tabular}
\end{table}

\subsection{Techniques} \label{sec:technique}

\paragraph{Upper bounds.} Since the trace of quantum state power \(\tr(\rho^q)\) is a unitarily invariant quantity, it is well-known that there exists a canonical estimator performing weak Schur sampling~\cite{CHW07,MdW16,OW21} on \(\rho^{\otimes n}\) to obtain a Young diagram outcome \(\lambda\) and 
then predicting the final result \(\tr(\rho^q)\) based on \(\lambda\). The most straightforward way to do this is to treat each \(\lambda_i/n\), where \(\lambda_i\) is the \(i\)-th row of \(\lambda\), as an estimate of the \(i\)-th large eigenvalue of \(\rho\), and then output \(\sum_{i}(\lambda_i/n)^q\) as the final result, which is what is called the \textit{plug-in estimator}. 
Existing quantum plug-in estimators are known for, e.g., von Neumann entropy and R\'enyi entropy in \cite{AISW20,BMW16}. 

However, directly using the plug-in estimator with current error bounds for weak Schur sampling in \cite{OW17} seems to be difficult to avoid the dependence on the dimension (or rank) of $\rho$ appearing in the accumulation of errors. 
This is very different from the classical empirical estimation. For example, the classical plug-in estimators for $\sum_i p_i^{q}$ in \cite{JVHW15,JVHW17} suffice to achieve the optimal sample complexity, while the same strategy might introduce an \textit{unexpected} factor of $\operatorname{poly}(d)$ in the quantum case, where $d$ is the dimension.
To overcome this limitation, we develop non-plug-in estimators for $\tr\rbra{\rho^q}$. 
Our non-plug-in estimator adopts a simple but effective truncation strategy which eliminates the dimension (or rank) in the complexity. 
Specifically, having obtained an estimated spectrum $\hat \alpha = \rbra{\hat\alpha_1, \hat\alpha_2, \dots, \hat\alpha_d}$ of $\rho$ to certain precision with $\hat\alpha_1 \geq \hat\alpha_2 \geq \dots \geq \hat\alpha_d$ (with $\alpha = \rbra{\alpha_1, \alpha_2, \dots, \alpha_d}$ the true sorted spectrum of $\rho$), our non-plug-in estimator is then of the form 
\[
    \hat P = \sum_{j=1}^m \hat \alpha_j^q,
\]
where $m$ is a truncation parameter such that the lower-order errors are controlled by the eigenvalues (which are finally suppressed due to constantly upper bounded partial sums), and the higher-order errors are accumulated with scaling only depending on $m$ (thus suppressed with negligible truncation bias).
In sharp contrast to the quantum plug-in estimators in the literature \cite{AISW20,BMW16}, our non-plug-in construction can be shown to achieve optimal sample complexity only up to a logarithmic factor (see \cref{sec:q>2,sec:1<q<2} for more details).
As a result, we obtain sample upper bounds $\widetilde{O}\rbra{1/\varepsilon^2}$ for $q > 2$, $\widetilde{O}\rbra{1/\varepsilon^{\frac{2}{q-1}}}$ for $1 < q < 2$, and $\widetilde{O}\rbra{d^{2/q}/\varepsilon^{2/q}}$ for $0 < q < 1$.
Note that the exponent of the upper bound does not depend on $q$ for constant $q > 2$, which is in contrast to $1 < q < 2$.  
This is because we borrow a factor $\alpha_i$ from $\alpha_i^q$ to control the error $|\hat{\alpha}_i-\alpha_i|$ (to avoid $d$-dependence), and the fluctuation of $\hat{\alpha}_i^{q-1}$ is small enough when $q>2$ (see \cref{eq:using-power-eq-1}), causing $q$ to disappear from the exponent.

\paragraph{Lower bounds.} 
The lower bound \(\Omega(1/\varepsilon^2)\) for any constant \(q>1\) is obtained by reducing from a state discrimination task with a simple but effective hard instance from \cite{CWLY23,GHYZ24}.

The lower bound \(\Omega(1/\varepsilon^{\frac{1}{q-1}})\) for \(1<q<2\) is obtained by reducing a state discrimination task on ensembles of quantum states. Specifically, we consider two unitarily invariant ensembles of quantum states that are maximally mixed with respect to different dimensions. 
Then, we show that the discrimination between these ensembles can be characterized by the discrimination between certain Schur--Weyl distributions in their total variation distance. 
To bound the total variation distance, we recall the relationship between the Schur--Weyl distributions and Plancherel distributions shown in \cite{CHW07}, which demands a linear scaling with the dimensions. With carefully chosen dimension parameters, we can obtain our lower bound.

The lower bound $\Omega(d^{1/q}/\varepsilon^{1/q})$ for $0<q<1$ can also be established via a reduction from a state discrimination task over ensembles of quantum states. Compared to the previous construction, we introduce an additional parameter $\Delta$ that controls the frequency of ``valid'' samples: among $n$ samples, only $O(\Delta n)$ of them (with high probability) contribute to the discrimination task. 
These ``valid'' samples exhibit unitary invariance and induce specific Schur--Weyl distributions. We can then upper bound the relevant trace distance by using the Plancherel distribution as an intermediate tool~\cite{CHW07}.

\paragraph{Computational hardness for $0<q<1$.} Our \NIQSZK{}-hardness results conceptually follow the framework of~\cite{CCKV08,KLGN19}, which treats the case $q=1$ corresponding to the von Neumann entropy. At a high level, this framework reduces \QSCMM{}, the problem of distinguishing an arbitrary efficiently preparable $n$-qubit state $\rho$ from the maximally mixed state $(I/2)^{\otimes n}$ with respect to the trace distance, to the promise problem \TsallisQEA{} of estimating the $q$-Tsallis entropy of $\rho$. Because $\rho$ and $(I/2)^{\otimes n}$ are simultaneously diagonalizable, the reduction relies on an inequality relating the $q$-Tsallis entropy of the eigenvalue distribution of $\rho$ to its total variation distance from the uniform distribution. 

Proving such an inequality leads to an optimization problem over the eigenvalue distribution. 
For the regime $1\leq q \leq 1+\frac{1}{n}$, the (possibly non-convex) optimization problems arising in the proofs of~\cite{KLGN19,LW25} admit explicit solutions.\footnote{As observed in~\cite{KLGN19}, when $q=1$, the upper bound on the Shannon entropy in this setting follows from Vajda's inequality~\cite{Vajda70}.} In fact, a direct analog of this approach is available only for $q=1/2$. In contrast, for the general regime $0<q<1$, obtaining an explicit solution is often challenging due to the non-convex nature of the optimization problem. Instead, the $q$-Tsallis entropy is upper bounded by analyzing a truncated Taylor expansion of the objective function together with a Lagrange remainder term, yielding a weaker bound (see \Cref{thm:inequality-uniformTV-TsallisEA}) that nevertheless suffices to establish the reduction $\QSCMM{} \leq \TsallisQEA$ and hence the \NIQSZK{}-hardness. 

\subsection{Related Work}

After the work of \cite{BCWdW01,EAO+02}, there have been a series of subsequent work focusing on the estimation of $\tr\rbra{\rho^q}$ for integer $q \geq 2$ \cite{Bru04,vEB12,JST17,SCC19,YS21,QKW24,ZL24,SLLJ24,YLLW24,CWYZ25,Wan25,SJW+25}. 
The quantum query complexity of entropy estimation has also been extensively studied in the literature, including the von Neumann entropy \cite{GL20,GHS21,WGL+24} and R\'enyi entropy \cite{SH21,WGL+24,WZL24}.
As the classical counterpart, estimating the functional $\sum_{j=1}^N p_j^q$ of a probability distribution $p$ to within an additive error was studied in \cite{AK01} for integer $q \geq 2$, and later in \cite{JVHW15,JVHW17} for non-integer $q$; its estimation to a multiplicative error was studied in \cite{AOST17} for R\'enyi entropy estimation. 
In addition, Shannon entropy estimation was studied in 
\cite{Pan03,Pan04,VV11a,VV11b,VV17,WY16}.

Given sample access to the quantum states to be tested, quantum estimators and testers for their properties have been investigated in the literature. 
The first optimal quantum tester was discovered in \cite{CHW07}, which distinguishes whether a quantum state has a spectrum uniform on $r$ or $2r$ eigenvalues. 
This was later generalized to an optimal tester for mixedness in \cite{OW21} and to quantum state certification in \cite{BOW19}. 
In addition, optimal estimators are known for R\'enyi entropy of integer order \cite{AISW20}, and the closeness (trace distance and fidelity) between pure quantum states \cite{WZ24b}. 
A distributed optimal estimator was known for the inner product of quantum states \cite{ALL22}.
Estimators and testers with incoherent measurements are also known for purity \cite{CCHL21,GHYZ24}, unitarity \cite{CCHL21,CWLY23}, certification \cite{CHLL22,LA24}, and $\tr\rbra{\rho^q}$ for integer $q$ (further used for spectrum estimation)~\cite{PTTW25}. 
In addition to those that were known to be optimal, there are also estimators for entropy \cite{AISW20,BMW16,WZ24,LW25}, relative entropy \cite{Hay24}, fidelity \cite{GP22}, and trace distance \cite{WZ24c}. 

\subsection{Discussion}

In this paper, we presented quantum estimators for estimating $\tr\rbra{\rho^q}$ for non-integer $q > 1$, significantly improving the prior approaches.
In particular, for $q > 2$, our estimators achieve optimal sample complexity only up to a logarithmic factor. 
Our (non-plug-in) estimators are directly constructed by weak Schur sampling with optimal sample complexity (although every estimator for unitarily invariant properties is known to imply a canonical estimator based on weak Schur sampling \cite[Lemma 20]{MdW16}), in addition to the (plug-in) optimal estimator for R\'enyi entropy of integer order \cite{AISW20}, the optimal testers for mixedness \cite{OW21} and quantum state certification \cite{BOW19}, and the optimal learners for full tomography \cite{HHJ+17,OW16}. 
At the end of the discussion, we list some questions in this direction for future research. 

\begin{enumerate}
    \item Can we remove the logarithmic factor from the sample complexity obtained in this paper?
    \item Can we improve the upper or the lower bound for $1 < q < 2$?
    \item Can we find more (plug-in or non-plug-in) optimal estimators based on weak Schur sampling?
    \item Can we obtain optimal estimators for $\tr\rbra{\rho^q}$ with restricted measurements?
\end{enumerate}

\section{Preliminaries}

\subsection{Notations}
In our quantum algorithms for the cases of $q>2$ and $1<q<2$, we use the following three parameters $m, \delta', \varepsilon'$, where $m$ is the position where the truncation is taken, and $\delta'$ and $\varepsilon'$ are, respectively, the failure probability and the precision when applying the quantum spectrum estimation with entry-wise bounds in \cref{sec:qse}. 
Specifically, $m \in \sbra{d}$ is a positive integer and $\delta', \varepsilon' \in \rbra{0, 1}$ are real numbers, all of which are to be determined later. 
For the case of $0<q<1$, we use the parameter $\eta$, which specifies the precision for estimating the spectrum of $\rho$ in the total variation distance.
Throughout this paper, we assume that $\rho$ has the spectrum decomposition:
\begin{equation*}
    \rho = \sum_{j=1}^d \alpha_j \ketbra{\psi_j}{\psi_j},
\end{equation*}
where $\alpha_1\geq \alpha_2\geq \cdots \geq \alpha_d \geq 0$ with $\sum_{j=1}^d \alpha_j = 1$ and $\cbra{\ket{\psi_j}}$ is an orthonormal basis. 

In addition, we adopt the notion of the \emph{$q$-logarithm} $\ln_q(x) \coloneqq \frac{x^{1-q}-1}{1-q}$, which converges to the natural logarithm $\ln(x)$ as $q\to 1$. This notion arises naturally in the definition of the quantum $q$-Tsallis entropy,  as reflected by the alternative expression $\Sq(\rho) = - \tr\rbra*{\rho^q \ln_q(\rho)}$.

For our purpose, we also need the following inequalities.

\begin{fact} \label{fact:power-ge-1}
    For $\alpha > 1$ and $x, y \in \sbra{0, 1}$, we have $x^{\alpha} \leq x$ and $\abs{x^\alpha - y^\alpha} \leq \alpha \abs{x - y}$. 
\end{fact}
\begin{proof}
    This fact follows by applying the mean value theorem on the function $f\rbra{x} = x^{\alpha}$.
\end{proof}

\begin{fact}\label{prop-12251837}
For \(0\leq x\leq y\leq 1\) and \(0< s< 1\) we have
$y^s-x^s\leq (y-x)^s$.
\end{fact}
\begin{proof}
This fact follows by considering the derivative of the function \(f(x)\coloneqq (y-x)^s +x^s\).
\end{proof}

\begin{fact}\label{prop-12251915}
For \(0<s<1\) and \( x_i\geq 0\) for all $1 \leq i \leq k$, we have
\[\sum_{i=1}^k x_i^s \leq k^{1-s}\cdot \left(\sum_{i=1}^k x_i\right)^s. \]
\end{fact}
\begin{proof}
Let \(y_i = x_i^s\). By Roger--H{\"o}lder's inequality \cite{Rog88,Hol89}, 
\begin{align}
\sum_{i=1}^k x_i^s &=\sum_{i=1}^k y_i \leq \left(\sum_{i=1}^k 1^{\frac{1}{1-s}}\right)^{1-s} \left(\sum_{i=1}^k y_i^{\frac{1}{s}}\right)^s = k^{1-s} \cdot \left(\sum_{i=1}^k x_i\right)^s. \nonumber 
\end{align}
\end{proof}

\subsection{Basics in quantum computing}
A $d$-dimensional (mixed) quantum state can be described by a $d \times d$ complex-valued positive semidefinite matrix $\rho \in \mathbb{C}^{d \times d}$ satisfying $\tr\rbra{\rho} = 1$. 
The trace distance between two quantum states $\rho_0$ and $\rho_1$ is defined by
\begin{equation*}
    \td(\rho_0,\rho_1) \coloneqq \frac{1}{2} \Abs*{\rho_0-\rho_1}_1 = \frac{1}{2} \tr\rbra*{\abs*{\rho_0-\rho_1}}.
\end{equation*}
The fidelity between two quantum states $\rho_0$ and $\rho_1$ is defined by
\begin{equation*}
    \mathrm{F}\rbra{\rho_0, \rho_1} \coloneqq \tr\rbra*{\sqrt{\sqrt{\rho_1}\rho_0\sqrt{\rho_1}}}.
\end{equation*} 
To discriminate two quantum states, we include the following well-known results. 
The following theorem can be found in {\cite[Section 9.1.4]{Wil13}}, {\cite[Lemma 3.2]{Hay16}}, and {\cite[Theorem 3.4]{Wat18}}.

\begin{theorem} [Quantum state discrimination, cf.\ {\cite[Section 9.1.4]{Wil13}}, {\cite[Lemma 3.2]{Hay16}}, and {\cite[Theorem 3.4]{Wat18}}] \label{thm:qsd}
    Any POVM $\Lambda = \cbra{\Lambda_0, \Lambda_1}$ that distinguishes two quantum states $\rho_0$ and $\rho_1$ (each with a priori probability $1/2$) with success probability
    \begin{equation*}
        \frac{1}{2} \tr\rbra{\Lambda_0 \rho_0} + \frac{1}{2} \tr\rbra{\Lambda_1 \rho_1} \leq \frac{1}{2}\rbra*{1 + \frac{1}{2}\Abs{\rho_0- \rho_1}_1}.
    \end{equation*}
\end{theorem}

The following fact was noted in \cite[Section 1]{HHJ+17}, which is closely related to the quantum Chernoff bound \cite{NS09,ACMT+07}.

\begin{fact} \label{fact:qsd}
    The sample complexity for distinguishing two quantum states $\rho_0$ and $\rho_1$ is $\Omega\rbra{1/\gamma}$,
    where $\gamma = 1 - \mathrm{F}\rbra{\rho_0, \rho_1}$ is the infidelity. 
\end{fact}

\subsection{Basic representation theory}

A \textit{representation} of a group \(G\) is a pair \((\mu,\mathcal{H})\), where \(\mathcal{H}\) is a (complex) Hilbert space, and \(\mu: G\rightarrow \textup{GL}(\mathcal{H})\) is a group homomorphism from \(G\) to the general linear group on \(\mathcal{H}\).\footnote{In this paper, we mostly consider the case that \(G\) is finite or compact, where without loss of generality we can assume \(\mu:G\rightarrow \mathbb{U}(\mathcal{H})\) is unitary.} 
We also call \(\mu(g)\) the action of \(g\in G\) on \(\mathcal{H}\). When the group action is clear from the context, we may omit \(\mu\) and directly use \(\mathcal{H}\) to refer to the representation of \(G\).

A \textit{sub-representation} of \((\mu,\mathcal{H})\) is a representation \((\mu',\mathcal{H}')\), where \(\mathcal{H}'\) is a subspace of \(\mathcal{H}\) and \(\mu'(g)\) is the restriction of \(\mu(g)\) to \(\mathcal{H}'\). 
A representation \(\mathcal{H}\) of \(G\) is \textit{irreducible} if the only sub-representations of \(\mathcal{H}\) are \(\{0\}\) and \(\mathcal{H}\) itself.
A \textit{representation homomorphism} between two representations \((\mu_1,\mathcal{H}_1),(\mu_2,\mathcal{H}_2)\) of group \(G\) is a linear operator \(F:\mathcal{H}_1\rightarrow \mathcal{H}_2\) which commutes with the action of \(G\), i.e., 
\begin{equation*}
F\mu_1(g)=\mu_2(g)F.
\end{equation*}
A \textit{representation isomorphism} is a representation homomorphism that is also a full-rank linear map. Two representations \(\mathcal{H}_1\) and \(\mathcal{H}_2\) of a group \(G\) are said to be \textit{isomorphic} if there exists a representation isomorphism between them, and we write \(\mathcal{H}_1\stackrel{G}{\cong}\mathcal{H}_2\). Then, we introduce the Schur's Lemma, which is an important and basic result in representation theory.
\begin{fact}[Schur's Lemma, see, e.g.~{\cite[Proposition 2.3.9]{etingof2011introduction}}] \label{prop:202404032229}
Let \(\mathcal{H}_1,\mathcal{H}_2\) be irreducible representations of a group \(G\). If \(F: \mathcal{H}_1\rightarrow \mathcal{H}_2\) is a non-zero homomorphism of representations, then \(F\) is an isomorphism.
\end{fact}
The following is a direct and useful corollary of Schur's Lemma.
\begin{corollary}\label{coro-240446}
Suppose \(\mathcal{H}\) is an irreducible representation of \(G\) and \(F:\mathcal{H}\rightarrow \mathcal{H}\) is a representation homomorphism. Then \(F=cI\) where \(c\) is a complex number. 
\end{corollary}

\subsubsection{Schur--Weyl duality}

A \textit{Young diagram} \(\lambda\) with \(n\) boxes and at most \(d\) rows is a partition \(\lambda=(\lambda_1,\ldots,\lambda_d)\) of \(n\) such that \(\sum_{i} \lambda_i=n\) and \(\lambda_1\geq\cdots \geq \lambda_d\geq 0\). 
For example, the Young diagram with \(8\) boxes and \(3\) rows, identified by the partition \((4,2,1)\) is:
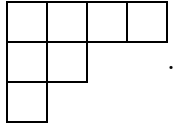
\begin{equation*}
\vcenter{\hbox{\scalebox{0.9}{\begin{ytableau}~&~&~&~\\~&~\\~\end{ytableau}}}}.
\end{equation*}
We use \(\lambda\vdash n\) to denote that \(\lambda\) is a Young diagram with \(n\) boxes.

Consider the actions of the symmetric group \(\mathfrak{S}_n\) and the unitary group \(\mathbb{U}_d\) on the Hilbert space \((\mathbb{C}^d)^{\otimes n}\). For any \(U\in\mathbb{U}_d\), \(U\) acts on \((\mathbb{C}^d)^{\otimes n}\) by
\[\ket{\psi_1}\otimes \cdots\otimes \ket{\psi_n} \mapsto U\ket{\psi_1}\otimes \cdots \otimes U\ket{\psi_n},\]
and for any \(\pi\in\mathfrak{S}_n\), \(\pi\) acts on \((\mathbb{C}^d)^{\otimes n}\) by
\begin{equation*}
\ket{\psi_1}\otimes \cdots\otimes \ket{\psi_n}\mapsto \ket{\psi_{\pi^{-1}(1)}}\otimes \cdots \otimes \ket{\psi_{\pi^{-1}(n)}}.
\end{equation*}
For convenience, we directly use \(U^{\otimes n}\) and \(\pi\) to denote the action of \(U\in\mathbb{U}_d\) and \(\pi\in\mathfrak{S}_n\) on \((\mathbb{C}^d)^{\otimes n}\).

Note that \(U^{\otimes n}\) and \(\pi\) commute with each other, which means \((\mathbb{C}^d)^{\otimes n}\) is also a representation of the group \(\mathfrak{S}_n\times \mathbb{U}_d\). 
This is characterized by the following renowned Schur--Weyl duality.
\begin{fact}[Schur--Weyl duality~\cite{fulton2013representation,etingof2011introduction}]\label{prop-432239}
\begin{equation*}
(\mathbb{C}^{d})^{\otimes n}\stackrel{\mathfrak{S}_n\times \mathbb{U}_d}{\cong}\bigoplus_{\lambda \vdash n}\mathcal{P}_\lambda\otimes \mathcal{Q}^d_\lambda,
\end{equation*}
where \(\mathcal{P}_\lambda\) and \(\mathcal{Q}_\lambda^d\) are irreducible representations of \(\mathfrak{S}_n\) and \(\mathbb{U}_d\), respectively, and are labeled by a Young diagram \(\lambda\vdash n\).\footnote{Note that if the Young diagram \(\lambda\) has more than \(d\) rows, then \(\mathcal{Q}_\lambda^d=0\).}
\end{fact}
For \(\pi\in \mathfrak{S}_n\) and \(U\in\mathbb{U}_d\), we use \(\mathtt{p}_{\lambda}(\pi)\) and \(\mathtt{q}_{\lambda}(U)\) to denote their actions on \(\mathcal{P}_\lambda\) and \(\mathcal{Q}_{\lambda}^d\), respectively. 
\begin{remark}\label{remark-241950}
In fact, \(\mathtt{q}_\lambda\) can be extended naturally to the actions of the group \(\textup{GL}(\mathbb{C}^d)\) on \(\mathcal{Q}_\lambda^d\), and further by continuity to the action of any matrix in \(\textup{End}(\mathbb{C}^d)\) on \(\mathcal{Q}_\lambda^d\).
\end{remark}
For any matrix \(X\in \textup{End}(\mathbb{C}^d)\), \(X^{\otimes n}\) is invariant under permutations (the actions of \(\mathfrak{S}_n\)). It is not hard using Schur's Lemma to show the following fact. 
\begin{fact}\label{fact-241853}
\(X^{\otimes n}\) has the following form:
\[X^{\otimes n}=\bigoplus_{\lambda\vdash n}  I_{\mathcal{P}_\lambda}\otimes \mathtt{q}_\lambda(X),\]
where \(\mathtt{q}_\lambda(X)\) is the action of \(X\) on \(\mathcal{Q}_\lambda^d\) (see \cref{remark-241950}).
\end{fact}
Furthermore, it is known that \(\tr(\mathtt{q}_{\lambda}(X))=s_\lambda(\alpha)\), where \(s_\lambda\) is the \textit{Schur polynomial}~\cite{fulton2013representation} indexed by \(\lambda\) and \(\alpha=(\alpha_1,\ldots,\alpha_d)\) are the eigenvalues of \(X\).

\subsection{Weak Schur sampling as quantum estimators}

Suppose we have \(n\) samples of an unknown \(d\)-dimensional quantum state \(\rho\). Consider the task of estimating a quantitative property \(F(\rho)\) of \(\rho\)  (e.g., the purity \(\tr(\rho^2)\)).
Without loss of generality, the estimator can be described by a POVM \(\{M_i\}\) applied on \(\rho^{\otimes n}\),\footnote{Here, we assume the POVM is discrete, the continuous case can be treated similarly.} and \(f(i)\) is returned as an estimate if the measurement outcome is \(i\).

Note that \(\rho^{\otimes n}\) is invariant under permutations of the tensors, i.e.,
for any \(\pi\in\mathfrak{S}_n\), \(\pi \rho^{\otimes n}\pi^\dag =\rho^{\otimes n}\).
This means we can ``factor out'' the action of the symmetric group \(\mathfrak{S}_n\) to obtain a permutation invariant estimator.
Furthermore, if the quantitative property \(F(\rho)\) is unitarily invariant, i.e.,
for any \(U\in \mathbb{U}_d\), \(F(U\rho U^\dag)= F(\rho)\),
we can also factor out the action of the unitary group \(\mathbb{U}_d\) to obtain a unitarily invariant estimator with the performance no worse than the original one. 
Specifically, we define the canonical permutation-invariant and unitary-invariant estimator \(\{\overline{M}_i\}\) as:
\[\overline{M}_i= \frac{1}{n!}\sum_{\pi\in\mathfrak{S}_n}\pi\E_{U\in \mathbb{U}_d}\left[ U^{\otimes n} M_i U^{\dag\otimes n}\right]\pi^\dag.\]
The following shows that the estimator \(\{\overline{M}_i\}_i\) is at least as powerful as the original estimator \(\{M_i\}_i\) (see also, e.g., \cite{MdW16,Hay24}).
\begin{fact}
If \(\{M_i\}\) is an estimator of the quantitative property \(F\) achieving additive error \(\varepsilon\) with success probability \(1-\delta\), then \(\{\overline{M}_i\}\) can also achieve additive error \(\varepsilon\) with probability \(1-\delta\).
\end{fact}

Note that \(\overline{M}_i\) commutes with both \(U^{\otimes n}\) and \(\pi\) for any \(U\in\mathbb{U}_d\) and \(\pi\in\mathfrak{S}_n\). 
By the Schur--Weyl duality (see \cref{prop-432239}) and \cref{coro-240446}, we have
\[\overline{M}_i=\bigoplus_{\lambda\vdash n} c_{i,\lambda}\cdot  I_{\mathcal{P}_\lambda}\otimes I_{\mathcal{Q}_\lambda^d},\]
where \(c_{i,\lambda}\) is a positive number such that \(\sum_{i}c_{i,\lambda}=1\).
Then, by \cref{fact-241853}, we can see that the estimator \(\{\overline{M}_i\}\) applied on \(\rho^{\otimes n}\) is equivalent to 
\begin{enumerate}
    \item sample a \(\lambda\vdash n\) from the distribution \(\{\tr(I_{\mathcal{P}_\lambda}\otimes \mathtt{q}_{\lambda}(\rho))\}_\lambda=\{\dim(\mathcal{P}_\lambda)\cdot s_\lambda(\alpha)\}_\lambda\), where \(s_\lambda\) is the Schur polynomial and \(\alpha=(\alpha_1,\ldots,\alpha_d)\) are the eigenvalues of \(\rho\) such that \(\alpha_1\geq \cdots \geq\alpha_d\).
    \item sample an \(i\) from the distribution \(\{c_{i,\lambda}\}_i\). 
\end{enumerate}
It is worth noting that, the second step is entirely classical, while the first step is a quantum measurement independent of the specific task, which is called \textit{weak Schur sampling}~\cite{CHW07}. In step \(1\), the distribution \(\{\dim(\mathcal{P}_\lambda)\cdot s_\lambda(\alpha)\}_\lambda\) is referred to as the \textit{Schur--Weyl distribution}~\cite{OW17} and is denoted by 
\(\operatorname{SW}^n(\alpha)\) or \(\operatorname{SW}^n(\rho)\). Specifically,
\begin{equation*}
    \Pr_{\lambda' \sim \operatorname{SW}^n(\alpha)}\sbra{\lambda' = \lambda} = \dim(\mathcal{P}_\lambda)\cdot s_\lambda(\alpha).
\end{equation*}
Furthermore, the Young diagram \(\lambda \sim \operatorname{SW}^n(\alpha)\) provides a good approximation of the eigenvalues \(\alpha_1,\ldots,\alpha_d\) of \(\rho\), which is characterized by the following results.

\begin{lemma}[{\cite[Theorem 1.7]{OW16}}]
It holds that
\[
\E_{\lambda \sim \operatorname{SW}^n\rbra*{\alpha}} \sbra*{ \frac{1}{2} \sum_{j \in \sbra{d}} \abs*{\frac{\lambda_j}{n} - \alpha_j} } \leq \frac{1.92 d + 0.5}{\sqrt{n}}.
\]
\end{lemma}

\begin{lemma}[Adapted from {\cite[Theorem 1.5]{OW17}}] \label{lemma:sw-2}
For $j \in \sbra{d}$, 
    \begin{equation*}
        \E_{\lambda \sim \operatorname{SW}^n\rbra*{\alpha}}\sbra*{\rbra*{\lambda_j - \alpha_j n}^2} \leq O\rbra{n}.
    \end{equation*}
\end{lemma}

We use \(\operatorname{SW}^n_d\) to denote \(\operatorname{SW}^n(\alpha)\) when \(\alpha=(1/d,\ldots,1/d)\),\footnote{In some papers, \(\operatorname{SW}_d^n\) is also called the Schur--Weyl distribution~\cite{OW21} or simply the Schur distribution~\cite{CHW07}.} i.e., \(\rho\) is maximally mixed. Furthermore, when \(d\rightarrow \infty\), the distribution tends to a limiting distribution \(\operatorname{Planch}(n)\), called \textit{Plancherel distribution} over the symmetric group \(\mathfrak{S}_n\). We will use the following result which provides both upper and lower bounds of the convergence of \(\operatorname{SW}_d^n\) to \(\operatorname{Planch}(n)\).

\begin{lemma}[{\cite[Lemma 6]{CHW07}}]\label{lemma-241736}
If \(n \leq d\), then
\[\|\operatorname{SW}_d^n-\operatorname{Planch}(n)\|_1\leq \sqrt{2}\frac{n}{d}.\]
\end{lemma}

\subsection{Quantum spectrum estimation with entry-wise bounds} \label{sec:qse}

Efficient approaches to quantum spectrum estimation were given in \cite{OW16} in the $\ell_1$ and $\ell_2$ distances and in \cite{OW17} in the Hellinger-squared distance, chi-squared divergence, and Kullback--Leibler (KL) divergence. 
In this section, we provide an efficient approach to quantum spectrum estimation with entry-wise bounds in \cref{algo:spectrum} based on the results of \cite{OW17}, which will be used as a subroutine in our estimators for $\tr\rbra{\rho^q}$.

\begin{algorithm}[h]
\caption{$\texttt{SpectrumEstimation}\rbra{\rho, n, k}$}
    \label{algo:spectrum}
    \begin{algorithmic}[1]
    \Require Sample access to a $d$-dimensional mixed quantum state $\rho$; integers $n, k \geq 1$. 

    \Ensure A $d$-dimensional vector $\hat\alpha \in \mathbb{R}^d$.

    \For {$l = 1, 2, \dots, k$}
    \State $\lambda^{\rbra{l}} \sim \textup{SW}^n\rbra{\rho}$. 
    \EndFor

    \For {$j = 1, 2, \dots, d$}
    \State $\hat\alpha_j \gets \operatorname{median}\cbra{\underline{\lambda}^{\rbra{1}}_j, \underline{\lambda}^{\rbra{2}}_j, \dots, \underline{\lambda}^{\rbra{k}}_j}$, where $\underline{\lambda}^{\rbra{l}}_j = \lambda_j^{\rbra{l}}/n$. 
    \EndFor

    \State \Return $\rbra{\hat\alpha_1, \hat\alpha_2, \dots, \hat\alpha_d}$.
    \end{algorithmic}
\end{algorithm}

For our purposes, we consider the case where $k > 1$, which is useful in our estimation algorithms for the case of $q > 1$ so that we can remove the dependence on $d$ in the complexity. 

\begin{lemma} [Quantum spectrum estimation with entry-wise bounds] \label{lemma:sample}
    For every $\varepsilon, \delta \in \rbra{0, 1}$, the process $\textup{\texttt{SpectrumEstimation}}\rbra{\rho, n, k}$ with $n = \Theta\rbra{1/\varepsilon^2}$ and $k = \Theta\rbra{\log\rbra{1/\delta}}$ uses $nk=O\rbra{\log\rbra{1/\delta}/\varepsilon^2}$ samples of $\rho$ and returns a sequence of random variables $\hat{\alpha} = \rbra{\hat{\alpha}_1, \hat{\alpha}_2, \dots, \hat{\alpha}_d} \in \mathbb{R}^d$ such that for every \(j\in [d]\), it holds with probability at least $1-\delta$ that
    $\abs{\hat{\alpha}_j - \alpha_j} \leq \varepsilon$. 
\end{lemma}

We also use the special case of \cref{algo:spectrum} with $n = \Theta\rbra{d/\varepsilon^2}$ and $k = 1$, which is actually the standard result on quantum state tomography in \cite{OW16}. 

\begin{lemma} [Quantum spectrum estimation in total variation distance, adapted from {\cite[Theorem 1.7]{OW16}}] \label{lemma:tv-estimation}
    For every $\varepsilon \in \rbra{0, 1}$, the process $\textup{\texttt{SpectrumEstimation}}\rbra{\rho, n, k}$ with $n = \Theta\rbra{d/\varepsilon^2}$ and $k = 1$ uses $nk=O\rbra{d/\varepsilon^2}$ samples of $\rho$ and returns a sequence of random variables $\hat{\alpha} = \rbra{\hat{\alpha}_1, \hat{\alpha}_2, \dots, \hat{\alpha}_d} \in \mathbb{R}^d$ such that it holds with probability at least $2/3$ that
    \[
    \sum_{j=1}^d \abs*{ \hat{\alpha}_j - \alpha_j } \leq \varepsilon. 
    \]
\end{lemma}

We also need Hoeffding's inequality. 

\begin{theorem} [Hoeffding's inequality, {\cite[Theorem 2]{Hoe63}}] \label{thm:hoeffding}
    Let $X_1, X_2, \dots, X_n$ be independent and identical random variables with $X_j \in \sbra{0, 1}$ for all $1 \leq j \leq n$.
    Then,
    \begin{equation*}
        \Pr\sbra*{ \abs*{\frac{1}{n} \sum_{j=1}^n X_j - \E\sbra*{X_1}} \leq t } \geq 1 - 2 \exp\rbra*{-2nt^2}.
    \end{equation*}
\end{theorem}

\begin{proof}[Proof of \cref{lemma:sample}]
    In this proof, all expectations are computed over $\lambda \sim \operatorname{SW}^n\rbra*{\alpha}$. 
    Let \(\underline{\lambda}_j=\lambda_j/n\). By \cref{lemma:sw-2}, we have
    \[\E\sbra*{(\underline{\lambda}_j-\alpha_j)^2}\leq \frac{c}{n}\]
    for some constant \(c > 0\). Therefore, 
    \begin{align}
    \Pr\left[|\underline{\lambda}_j-\alpha_j|\geq  2\sqrt{\frac{c}{n}}\,\right]\cdot 4\cdot \frac{c}{n} 
    &\leq\E\sbra*{(\underline{\lambda}_j-\alpha_j)^2} \label{eq-230150}\\
    &\leq \frac{c}{n},\nonumber
    \end{align}
    where \cref{eq-230150} is by Markov's inequality that $\Pr\sbra{\abs{X} \geq a} \cdot a^k \leq \E\sbra{\abs{X}^k}$ for any random variable $X$, integer $k \geq 1$, and $a > 0$.
    This implies
    \[\Pr\left[|\underline{\lambda}_j-\alpha_j|\geq  2\sqrt{\frac{c}{n}}\,\right]\leq \frac{1}{4}.\]

    Now we draw $k$ independent samples $\lambda^{(1)}, \lambda^{(2)}, \dots, \lambda^{(k)}$ from $\operatorname{SW}^n\rbra*{\alpha}$, and for each $j \in \sbra{d}$ let
    \begin{equation*}
        \hat{\alpha}_j = \operatorname{median}\cbra*{\underline{\lambda}^{\rbra{1}}_j, \underline{\lambda}^{\rbra{2}}_j, \dots, \underline{\lambda}^{\rbra{k}}_j}.
    \end{equation*}
    Let $X_j^{\rbra{l}} \in \cbra{0, 1}$ be a random variable such that $X_j^{\rbra{l}} = 1$ if $\abs{\underline{\lambda}_j^{\rbra{l}} - \alpha_j} \geq 2\sqrt{c/n}$ and $0$ otherwise. 
    By Hoeffding's inequality (\cref{thm:hoeffding}) with $t = 1/12$, we have
    \begin{equation*}
        \Pr\sbra*{ \abs*{\frac{1}{k}\sum_{l=1}^k X_j^{\rbra{l}} - \E\sbra*{X_j^{\rbra{1}}}} \leq \frac{1}{12} } \geq 1 - 2 \exp\rbra*{-\frac{k}{72}}.
    \end{equation*}
    Note that $\E\sbra{X_j^{\rbra{1}}} \leq 1/4$, then
    \begin{equation*}
        \Pr\sbra*{ \frac{1}{k}\sum_{l=1}^k X_j^{\rbra{l}} \leq \frac{1}{3} } \geq 1 - 2 \exp\rbra*{-\frac{k}{72}},
    \end{equation*}
    which means that $\hat{\alpha}_j$, the median of $\underline{\lambda}^{\rbra{1}}_j, \underline{\lambda}^{\rbra{2}}_j, \dots, \underline{\lambda}^{\rbra{k}}_j$, satisfies $\abs{\hat{\alpha}_j - \alpha_j} \leq 2\sqrt{c/n}$ with probability
    \begin{equation*}
        \Pr \sbra*{ \abs*{\hat{\alpha}_j - \alpha_j} \leq 2\sqrt{\frac{c}{n}} \,} \geq 1 - 2 \exp\rbra*{-\frac{k}{72}}.
    \end{equation*}
    By taking $n = \ceil{4c/{\varepsilon^2}}$ and $k = \ceil{72\ln\rbra{2/\delta}}$, we have that for any $j \in \sbra{d}$,
    \begin{equation*}
        \Pr \sbra*{ \abs*{\hat{\alpha}_j - \alpha_j} \leq \varepsilon } \geq 1 - \delta.
    \end{equation*}
    The whole process uses $nk = O\rbra{\log\rbra{1/\delta}/\varepsilon^2}$ samples of $\rho$. 
\end{proof}

\subsection{Computational hardness of \QSCMM{}}

We begin by defining the promise problem \QSCMM{}:

\begin{definition}[\textsc{Quantum State Closeness to Maximally Mixed State}, {$\QSCMM[\beta,\alpha]$}, adapted from~{\cite[Section 3]{Kobayashi03}}]
    \label{def:QSCMM}
    Let $Q$ be a quantum circuit acting on $m$ qubits with $n$ specified output qubits, satisfying $m \geq 2n$.\footnote{An $m(n)$-qubit maximally mixed state $(I/2)^{\otimes n}$ can be prepared by a quantum circuit that generates $m(n)$ EPR pairs on $m=2n$ qubits and traces out one half of the qubits.} Let $\rho$ denote the quantum state obtained by applying $Q$ to the initial state $\ket{0}^{\otimes m}$ and tracing out the non-output qubits. Let $\alpha(n)$ and $\beta(n)$ be efficiently computable functions. The task is to decide whether the following holds: 
    \begin{itemize}
        \item \textbf{\emph{Yes:}} $\td\rbra*{\rho,(I/2)^{\otimes n}} \leq \beta(n)$; 
        \item \textbf{\emph{No:}} $\td\rbra*{\rho,(I/2)^{\otimes n}} \geq \alpha(n)$.
    \end{itemize}
\end{definition}

The problem \QSCMM{} provides a complete characterization of the class \NIQSZK{}~\cite{Kobayashi03,BASTS10,CCKV08}, which consists of promise problems admitting non-interactive quantum statistical zero-knowledge proofs. It is widely believed that $\BQP \subsetneq \NIQSZK$.
We next restate the corresponding \NIQSZK{}-hardness result, which will later be used to establish the computational hardness of estimating the quantum $q$-Tsallis entropy for $0<q<1$ in \Cref{subsec:comp-hardness-0leQle1}:
\begin{lemma}[\QSCMM{} is \NIQSZK{}-hard, adapted from~{\cite[Section 8.1]{CCKV08}}]
    \label{QSCMM-is-NIQSZKhard}
    For all integers $n\geq 3$, $\QSCMM{}\sbra*{\frac{1}{n}, 1-\frac{1}{n}}$ is \NIQSZK{}-hard. 
\end{lemma}

\section{The Case of \texorpdfstring{$q > 2$}{q > 2}}

\subsection{Upper bound}\label{sec:q>2}

\begin{algorithm}[h]
\caption{$\texttt{PowerTrace}\rbra{\rho, q, \varepsilon}$ for $q > 2$}
    \label{algo:q>2}
    \begin{algorithmic}[1]
    \Require Sample access to a $d$-dimensional mixed quantum state $\rho$; $q \in \rbra{2, +\infty}$ and $\varepsilon \in \rbra{0, 1}$. 

    \Ensure An estimate of $\tr\rbra{\rho^q}$.
    
    \State $\varepsilon' \gets \varepsilon/\rbra{q+3}$, $m \gets \min\cbra{\ceil{1/\varepsilon'}, d}$, $\delta' \gets 1/3m$. 

    \State $n \gets \Theta\rbra{1/\varepsilon'^2}$, $k \gets \Theta\rbra{\log\rbra{1/\delta'}}$. 

    \State $\hat\alpha \gets \texttt{SpectrumEstimation}\rbra{\rho, n, k}$. 

    \State $\hat P \gets \sum_{j=1}^{m} \hat\alpha_j^q$.

    \State \Return $\hat P$.
    \end{algorithmic}
\end{algorithm}

For $q > 2$, the sample complexity of estimating $\tr\rbra{\rho^q}$ is given as follows.

\begin{theorem} \label{thm:ub-q>2}
    For every constant $q > 2$, we can estimate $\tr\rbra{\rho^q}$ to within additive error $\varepsilon$ using $O\rbra{\log\rbra{1/\varepsilon}/\varepsilon^{2}}$ samples of $\rho$.
\end{theorem}

\begin{proof}
Our estimator for \cref{thm:ub-q>2} is formally given in \cref{algo:q>2}, where we set
\begin{equation}\label{eq-2100419}
\varepsilon' \coloneqq \frac{\varepsilon}{q+3},  \quad m=\min\cbra*{\ceil*{\frac{1}{\varepsilon'}}, d}, \quad\delta'\coloneqq \frac{1}{3m}, \quad n = \Theta\rbra*{\frac{1}{\varepsilon'^2}}, \quad k = \Theta\rbra*{\log\rbra*{\frac{1}{\delta'}}}.
\end{equation}
By \cref{lemma:sample}, we can use $nk = O\rbra{\log\rbra{1/\delta'}/\varepsilon'^2}$ samples of $\rho$ to obtain a sequence $\hat{\alpha} = \rbra{\hat{\alpha}_1, \hat{\alpha}_2, \dots, \hat{\alpha}_d}$ such that for every \(j\in [d]\),
\begin{equation} \label{eq:prob-diff-lambda-alpha}
    \Pr\sbra*{ \abs*{\hat{\alpha}_j - \alpha_j} \leq \varepsilon' } \geq 1 - \delta'. 
\end{equation}
Then, we consider the estimator:
\begin{equation*}
    \hat P \coloneqq \sum_{j=1}^m \hat\alpha_j^q.
\end{equation*}
The additive error is bounded by
\begin{align}
    \abs*{ \hat P - \tr\rbra{\rho^q} }
    & = \abs*{ \sum_{j=1}^m \rbra*{ \hat\alpha_j^q - \alpha_j^q } - \sum_{j=m+1}^d \alpha_j^q } \nonumber \\
    & \leq \sum_{j=1}^m \abs*{ \hat\alpha_j^q - \alpha_j^q } + \sum_{j=m+1}^d \alpha_j^q. \label{eq:diff-P-tr}
\end{align}

For the first term of \cref{eq:diff-P-tr}, note that
\begin{align}
    \hat\alpha_j^q - \alpha_j^q 
    & = \rbra{\hat\alpha_j - \alpha_j} \hat\alpha_j^{q-1} + \alpha_j \hat\alpha_j^{q-1} - \alpha_j^q \nonumber\\
    & = \rbra{\hat\alpha_j - \alpha_j} \hat\alpha_j^{q-1} + \alpha_j \rbra*{ \hat\alpha_j^{q-1} - \alpha_j^{q-1} }, \nonumber
\end{align}
then we have
\begin{align}
    \abs*{\hat\alpha_j^q - \alpha_j^q }
    & \leq \abs*{\hat\alpha_j - \alpha_j} \abs*{\hat\alpha_j}^{q-1} + \abs*{\alpha_j} \abs*{ \hat\alpha_j^{q-1} - \alpha_j^{q-1} } \nonumber \\
    & \leq \abs*{\hat\alpha_j - \alpha_j} \abs*{\hat\alpha_j} + \abs*{\alpha_j} \rbra{q-1} \abs*{\hat\alpha_j - \alpha_j} \label{eq:using-power-eq-1} \\
    & \leq \abs*{\hat\alpha_j - \alpha_j} \rbra*{ \abs*{\alpha_j} + \abs*{\hat\alpha_j - \alpha_j} } + \rbra{q-1} \abs*{\alpha_j} \abs*{\hat\alpha_j - \alpha_j} \nonumber \\
    & = q \alpha_j \abs*{\hat\alpha_j - \alpha_j} + \abs*{\hat\alpha_j - \alpha_j}^2, \label{eq:diff-alpha-hat}
\end{align}
where \cref{eq:using-power-eq-1} is by \cref{fact:power-ge-1}. 
From \cref{eq:diff-alpha-hat} and by \cref{eq:prob-diff-lambda-alpha}, the following holds with probability $\geq 1 - \delta'$:
\begin{equation*}
    \abs*{\hat\alpha_j^q - \alpha_j^q } \leq q \alpha_j \varepsilon' + \varepsilon'^2.
\end{equation*}
Therefore, we have that with probability $\geq 1 - m\delta'$, the following holds:
\begin{align}
    \sum_{j=1}^m \abs*{ \hat\alpha_j^q - \alpha_j^q } 
    & \leq \sum_{j=1}^m \rbra*{q \alpha_j \varepsilon' + \varepsilon'^2} \nonumber \\
    & = q\varepsilon' \sum_{j=1}^m \alpha_j + m \varepsilon'^2 \nonumber\\
    & \leq q \varepsilon' + m\varepsilon'^2. \label{eq:first-term}
\end{align}
On the other hand, by noting that $\alpha_j \leq 1/j$ (since $j \alpha_j\leq \alpha_1+\cdots+\alpha_j\leq 1$) for every $j \in \sbra{d}$, we have
\begin{align}
    \sum_{j=m+1}^d \alpha_j^q
    & \leq \sum_{j=m+1}^d \rbra*{\frac{1}{j}}^q \nonumber \\
    & \leq \int_{m}^d \rbra*{\frac{1}{x}}^q \mathrm{d} x \nonumber \\
    & = \frac{m^{1-q}-d^{1-q}}{q-1}. \label{eq:second-term}
\end{align}

Combining \cref{eq:first-term,eq:second-term} in \cref{eq:diff-P-tr}, we have that with probability $\geq 1 - m\delta'$, the following holds:
\begin{align}
    \abs*{ \hat P - \tr\rbra{\rho^q} } \leq q\varepsilon' + m\varepsilon'^2 + \frac{m^{1-q}-d^{1-q}}{q-1}. \label{eq:error-bound}
\end{align}
According to the choice of $\varepsilon', m, \delta'$ in \cref{eq-2100419} and by \cref{eq:error-bound}, we have that with probability $\geq 1 - m\delta' = 2/3$, it holds that
\begin{equation*}
    \abs*{ \hat P - \tr\rbra{\rho^q} } \leq \varepsilon .
\end{equation*}
To see this, we consider the following two cases:
\begin{enumerate}
    \item $1/\varepsilon' \leq d$. In this case, $1/\varepsilon' \leq m = \ceil{1/\varepsilon'} < 1/\varepsilon' + 1$. We have
    \begin{align}
        \eqref{eq:error-bound} 
        & \leq q\varepsilon' + \rbra*{\frac{1}{\varepsilon'} + 1} \varepsilon'^2 + \frac{1}{m} \nonumber\\
        & \leq q\varepsilon' + \varepsilon' + \varepsilon'^2 + \varepsilon' \nonumber \\
        & \leq \rbra{q+3}\varepsilon' \nonumber \\
        & = \varepsilon. \nonumber
    \end{align}
    \item $1/\varepsilon' > d$. In this case, $m = d < 1/\varepsilon'$. We have
    \begin{align}
        \eqref{eq:error-bound} 
        & = q\varepsilon' + d \varepsilon'^2 \nonumber \\
        & \leq \rbra{q+1}\varepsilon' \nonumber \\
        & < \varepsilon. \nonumber
    \end{align}
\end{enumerate}

To complete the proof, the sample complexity is
\begin{equation*}
    O\rbra*{\frac{\log\rbra{1/\delta'}}{\varepsilon'^2}} = O\rbra*{\frac{\log\rbra{1/\varepsilon}}{\varepsilon^{2}}}.
\end{equation*}
\end{proof}

\subsection{Lower bound}
\begin{theorem} \label{thm:lb}
    For any constant $q > 1$, any quantum estimator to additive error $\varepsilon$ for $\tr\rbra{\rho^q}$ requires sample complexity $\Omega\rbra{1/\varepsilon^2}$.
\end{theorem}
\begin{proof}
    Consider the problem of distinguishing the two quantum states $\rho_{\pm}$, where
    \begin{equation*}
        \rho_{\pm} = \rbra*{\frac{2}{3} \pm \varepsilon} \ketbra{0}{0} + \rbra*{\frac{1}{3} \mp \varepsilon} \ketbra{1}{1}.
    \end{equation*}
    Then,
    \begin{equation*}
        \tr\rbra{\rho_\pm^q} = \rbra*{\frac{2}{3} \pm \varepsilon}^q + \rbra*{\frac{1}{3} \mp \varepsilon}^q. 
    \end{equation*}
    \begin{equation*}
        \tr\rbra{\rho_+^q} - \tr\rbra{\rho_-^q}
        = \rbra*{ \rbra*{\frac{2}{3} + \varepsilon}^q - \rbra*{\frac{2}{3} - \varepsilon}^q } + \rbra*{ \rbra*{\frac{1}{3} - \varepsilon}^q - \rbra*{\frac{1}{3} + \varepsilon}^q }.
    \end{equation*}
    By the direct calculation that
    \begin{equation*}
        \lim_{\varepsilon \to 0} \frac{ \tr\rbra{\rho_+^q} - \tr\rbra{\rho_-^q} }{\varepsilon} = 2q \rbra*{ \rbra*{\frac{2}{3}}^{q-1} - \rbra*{\frac{1}{3}}^{q-1} } = \Theta\rbra{1},
    \end{equation*}
    we conclude that
    \begin{equation*}
        \tr\rbra{\rho_+^q} - \tr\rbra{\rho_-^q} = \Theta\rbra{\varepsilon}.
    \end{equation*}
    Therefore, any quantum estimator for $\tr\rbra{\rho^q}$ to additive error $\Theta\rbra{\varepsilon}$ can be used to distinguish $\rho_+$ and $\rho_-$. 
    On the other hand, if the quantum estimator for $\tr\rbra{\rho^q}$ to additive error $\varepsilon$ has sample complexity $S$, then $S \geq \Omega\rbra{1/\gamma}$.
    A direct calculation shows that the infidelity
    \begin{equation*}
        \gamma = 1 - \mathrm{F}\rbra{\rho_+, \rho_-} =  1 - \rbra*{ \sqrt{\frac{4}{9} - \varepsilon^2} + \sqrt{\frac{1}{9} - \varepsilon^2} } = \Theta\rbra{\varepsilon^2}. 
    \end{equation*}
    By \cref{fact:qsd}, we have $S = \Omega\rbra{1/\varepsilon^2}$. 
\end{proof}

\section{The Case of \texorpdfstring{$1 < q < 2$}{1 < q < 2}}

\subsection{Upper bound} \label{sec:1<q<2}

\begin{algorithm}[h]
\caption{$\texttt{PowerTrace}\rbra{\rho, q, \varepsilon}$ for $1 < q < 2$}
    \label{algo:1<q<2}
    \begin{algorithmic}[1]
    \Require Sample access to a $d$-dimensional mixed quantum state $\rho$; $q \in \rbra{1, 2}$ and $\varepsilon \in \rbra{0, 1}$. 

    \Ensure An estimate of $\tr\rbra{\rho^q}$.
    
    \State $\varepsilon' \gets \rbra{\varepsilon/5}^{\frac{1}{q-1}}$, $m \gets \min\cbra{\ceil{1/\varepsilon'}, d}$, $\delta' \gets 1/3m$. 

    \State $n \gets \Theta\rbra{1/\varepsilon'^2}$, $k \gets \Theta\rbra{\log\rbra{1/\delta'}}$. 

    \State $\hat\alpha \gets \texttt{SpectrumEstimation}\rbra{\rho, n, k}$. 

    \State $\hat P \gets \sum_{j=1}^{m} \hat\alpha_j^q$.

    \State \Return $\hat P$.
    \end{algorithmic}
\end{algorithm}

We state the sample complexity of estimating $\tr\rbra{\rho^q}$ for the case of $1 < q < 2$ as follows.

\begin{theorem} \label{thm:ub-1<q<2}
    For every constant $1 < q < 2$, we can estimate $\tr\rbra{\rho^q}$ to within additive error $\varepsilon$ using $O\rbra{\log\rbra{1/\varepsilon}/\varepsilon^{\frac{2}{q-1}}}$ samples of $\rho$.
\end{theorem}

To show \cref{thm:ub-1<q<2}, we need the following inequalities.

\begin{lemma}\label{lemma-12252035}
Suppose that \(1<q<2\) and \(x_1\geq x_2\geq \cdots \geq x_N\geq 0\) with \(\sum_{i=1}^N x_i=1\). For any positive integer \(m\leq N\), we have
\[\sum_{i=m+1}^N x_i^q \leq \frac{1}{m^{q-1}}.\]
\end{lemma}
\begin{proof}

Note that if \(x_i\geq x_j\) and \(0\leq \Delta\leq x_j\), then it is easy to verify that
\[x_i^q+x_j^q \leq (x_i+\Delta)^q + (x_j-\Delta)^q.\]
For any sequence \(x_{m}\geq x_{m+1} \cdots \geq x_N\geq 0\), we define a new sequence by the following process: 
\begin{enumerate}
\item Find the smallest index \(j\) such that \(x_j< x_{m}\), and then find the largest index \(k\) such that \(k>j\) and \(x_k>0\). If there are no such \(j,k\), then do nothing.
\item Upon the success of finding \(j,k\), we define the new sequence by \(x'_i=x_i\) for all \(i\neq j,k\) and 
\[x'_j=x_j+\Delta, \quad x'_k=x_k-\Delta,\]
where \(\Delta=\min\{x_{m}-x_j, x_k\}\).
\end{enumerate}
It is obvious that
\[x_{m+1}^q+\cdots+x_N^q \leq (x'_{m+1})^q+\cdots+(x'_N)^q.\]
Starting from a sequence \(x_{m}\geq x_{m+1} \geq \cdots \geq x_N\), we define \(A=\sum_{i=m+1}^N x_i\). Then, we iteratively apply this process and finally get a sequence like
\[x_m, \underbrace{x_{m}, x_{m},\ldots, x_{m}}_{l}, y, \]
where \(l=\lfloor A/x_{m}\rfloor\) and \(y=A-l\cdot x_{m}\). Therefore
\begin{align}
\sum_{i=m+1}^{N} x_i^q
&\leq l\cdot x_{m}^q + y^q  \nonumber \\
&= x_{m+1}^q  \left(l+ \left(\frac{A}{x_{m}}-l\right)^q \right) \nonumber \\
&\leq x_{m}^q \cdot \frac{A}{x_{m}}\label{eq-12252137}\\
& \leq x_{m}^{q-1},\label{eq-12252145}
\end{align}
where \cref{eq-12252137} is because \(A/x_{m}-l < 1\). Then, by noting that \(x_{m}\leq 1/m\), we have
\[\eqref{eq-12252145}\leq \frac{1}{m^{q-1}}.\]
\end{proof}

Now we are ready to prove \cref{thm:ub-1<q<2}. 

\begin{proof} [Proof of \cref{thm:ub-1<q<2}]
Our estimator for \cref{thm:ub-1<q<2} is formally given in \cref{algo:1<q<2}, where we set
\begin{equation}\label{eq-2100426}
\varepsilon' \coloneqq \rbra*{\frac{\varepsilon}{5}}^{\frac{1}{q-1}},  \quad m=\min\cbra*{\ceil*{\frac{1}{\varepsilon'}}, d}, \quad\delta'\coloneqq \frac{1}{3m}, \quad n = \Theta\rbra*{\frac{1}{\varepsilon'^2}}, \quad k = \Theta\rbra*{\log\rbra*{\frac{1}{\delta'}}}.
\end{equation} 
By \cref{lemma:sample}, we can use $nk=O\rbra{\log\rbra{1/\delta'}/\varepsilon'^2}$ samples of $\rho$ to obtain a sequence $\hat{\alpha} = \rbra{\hat{\alpha}_1, \hat{\alpha}_2, \dots, \hat{\alpha}_d}$ such that for every \(j\in [d]\),
\begin{equation} \label{eq-18521225}
    \Pr\sbra*{ \abs*{\hat{\alpha}_j - \alpha_j} \leq \varepsilon' } \geq 1 - \delta'. 
\end{equation}
Then, we consider the estimator:
\begin{equation*}
    \hat{P} \coloneqq \sum_{j=1}^m \hat{\alpha}_j^q.
\end{equation*}
We have
\begin{align}
\left|\hat{P}-\tr(\rho^q) \right| 
&= \left|\sum_{j=1}^m \left(\hat{\alpha}_j^q -\alpha_j^q\right) - \sum_{j=m+1}^d \alpha_j^q \right| \nonumber \\
&\leq \sum_{j=1}^m \left|\hat{\alpha}_j^q-\alpha_j^q\right| + \sum_{j=m+1}^d \alpha_j^q. \label{eq-12251937}
\end{align}
For the first term of \cref{eq-12251937}, note that
\begin{align}
\left|\hat{\alpha}_j^q - \alpha_j^q\right|
&= \left|(\hat{\alpha}_j-\alpha_j)\hat{\alpha}_j^{q-1}+\alpha_j \rbra{\hat{\alpha}_j^{q-1}-\alpha_j^{q-1}}\right| \nonumber \\
& \leq \abs*{\hat{\alpha}_j-\alpha_j}\hat{\alpha}_j^{q-1}+\alpha_j \abs*{\hat{\alpha}_j^{q-1}-\alpha_j^{q-1}} \nonumber \\
&\leq \abs*{\hat{\alpha}_j-\alpha_j}\hat{\alpha}_j^{q-1}+\alpha_j \abs*{\hat{\alpha}_j-\alpha_j}^{q-1},\label{eq-12251831}
\end{align}
where the last inequality is by \cref{prop-12251837}. Then, by \cref{eq-18521225}, with probability \(\geq 1-\delta'\), the following holds:
\[\eqref{eq-12251831}\leq \varepsilon' (\alpha_j+\varepsilon')^{q-1}+\alpha_j (\varepsilon')^{q-1}.\]
This implies, with probability \(\geq 1-m\delta'\), we have
\begin{align}
\sum_{j=1}^m \abs*{\hat{\alpha}_j^q-\alpha_j^q}
&\leq \varepsilon' \sum_{j=1}^m (\alpha_j+\varepsilon')^{q-1}+(\varepsilon')^{q-1}\sum_{j=1}^m \alpha_j \nonumber \\
&\leq \varepsilon' \sum_{j=1}^m (\alpha_j+\varepsilon')^{q-1}+(\varepsilon')^{q-1} \nonumber \\
&\leq \varepsilon' m^{2-q} \cdot \left(m\varepsilon' + \sum_{j=1}^m \alpha_j\right)^{q-1} + (\varepsilon')^{q-1}\label{eq-12251933}\\
& \leq \varepsilon' m^{2-q}  \rbra*{m\varepsilon' + 1}^{q-1} + (\varepsilon')^{q-1}, \label{eq:sum-alpha-diff}
\end{align}
where \cref{eq-12251933} is by \cref{prop-12251915}.

Combining \cref{eq:sum-alpha-diff} with \cref{eq-12251937}, we have that, with probability $\geq 1 - m\delta'$, it holds that
\begin{equation} \label{eq:1<q<2-err}
    \left|\hat{P}-\tr(\rho^q)\right|\leq \varepsilon' m^{2-q}  \rbra*{m\varepsilon' + 1}^{q-1} + (\varepsilon')^{q-1} + \sum_{j=m+1}^d\alpha_j^q. 
\end{equation}

According to the choice of $\varepsilon', m, \delta'$ in \cref{eq-2100426} and by \cref{eq:1<q<2-err}, we have that with probability $\geq 1 - m\delta' = 2/3$, it holds that
\begin{equation*}
    \abs*{ \hat P - \tr\rbra{\rho^q} } \leq \varepsilon .
\end{equation*}
To see this, we consider the following two cases:
\begin{enumerate}
    \item $1/\varepsilon' \leq d$. In this case, $1/\varepsilon' \leq m = \ceil{1/\varepsilon'} < 2/\varepsilon'$. 
    We use \cref{lemma-12252035} to obtain:
\begin{equation}
\sum_{j=m+1}^d\alpha_j^q\leq \frac{1}{m^{q-1}}.\label{eq-12252001}
\end{equation}
    Using \cref{eq-12252001}, we have
    \begin{align}
        \eqref{eq:1<q<2-err} 
        & \leq \varepsilon' (2/\varepsilon')^{2-q}(2+1)^{q-1}+(\varepsilon')^{q-1}+(1/\varepsilon')^{1-q} \nonumber\\
        & \leq 3(\varepsilon')^{q-1} + 2(\varepsilon')^{q-1}\nonumber \\
        & \leq 5(\varepsilon')^{q-1} \nonumber \\
        & = \varepsilon. \nonumber
    \end{align}
    \item $1/\varepsilon' > d$. In this case, $m = d < 1/\varepsilon'$. We have
    \begin{align}
        \eqref{eq:1<q<2-err} 
        & \leq \varepsilon' \rbra{1/\varepsilon'}^{2-q}  \rbra*{1 + 1}^{q-1} + (\varepsilon')^{q-1} \nonumber \\
        & \leq 5 \rbra{\varepsilon'}^{q-1} \nonumber \\
        & = \varepsilon. \nonumber
    \end{align}
\end{enumerate}
To complete the proof, the sample complexity is 
\begin{equation*}
    O\rbra*{\frac{\log\rbra{1/\delta'}}{\varepsilon'^2}} = O\rbra*{\frac{\log\rbra{1/\varepsilon}}{\varepsilon^{\frac{2}{q-1}}}}.
\end{equation*}

\end{proof}

\subsection{Lower bound}

\begin{theorem} \label{thm:lb1to2}
    For every constant $1 < q < 2$, when the dimension of $\rho$ is sufficiently large, any quantum estimator to additive error $\varepsilon$ for $\tr\rbra{\rho^q}$ requires sample complexity $\Omega\rbra{1/\varepsilon^{\frac{1}{q-1}}}$.
\end{theorem}

\begin{proof}
Given integers \(1\leq r\leq d\), we use \(D_{r,d}\) to denote the \(d\times d\) diagonal matrix:
\[D_{r,d}\coloneqq \diag(\underbrace{\tfrac{1}{r},\ldots,\tfrac{1}{r} }_{r},\underbrace{0,\ldots,0}_{d-r}).\]
Let 
\begin{equation*}
r=\left\lfloor \frac{1}{(2\varepsilon)^{\frac{1}{q-1}}}\right\rfloor\quad \textup{ and }\quad d=\left\lfloor \frac{1}{\varepsilon^{\frac{1}{q-1}}}\right\rfloor+1.
\end{equation*}
If the number of samples \(n>r\), then we directly have \(n\geq \Omega(1/\varepsilon^{\frac{1}{q-1}})\). Therefore, we assume 
\begin{equation}\label{eq-242207}
n\leq r.
\end{equation}

Then, consider the following distinguishing problem.
\begin{problem}\label{prob-220209}
Suppose a \(d\)-dimensional quantum state \(\rho\) is in one of the  following with equal probability:
\begin{enumerate}
    \item \(\rho=\rho_1\coloneqq U D_{r,d} U^\dag\), where \(U\sim \mathbb{U}_d\) is a \(d\)-dimensional Haar random unitary.
    \item \(\rho=\rho_2\coloneqq D_{d,d}\). 
\end{enumerate}
The task is to distinguish between the above two cases.
\end{problem}

Note that \(\tr(\rho_1^{q})=1/r^{q-1}\geq 2\varepsilon\) and \(\tr(\rho_2^q)=1/d^{q-1}\leq \varepsilon\). Therefore, any estimator of \(\tr(\rho^{q})\) to additive error \(\frac{1}{2}\varepsilon = \Theta\rbra{\varepsilon}\) is able to distinguish the two cases in \cref{prob-220209}.

On the other hand, suppose we have \(n\) samples of \(\rho\). Then, for the first case (i.e., \(\rho=\rho_1\)), we have
\begin{align}
\E\sbra*{\rho_1^{\otimes n}}&=\E_{U\sim \mathbb{U}_d}\sbra*{U^{\otimes n} D_{r,d}^{\otimes n} U^{\dag\otimes n}} \nonumber\\
&=\bigoplus_{\lambda\vdash n} I_{\mathcal{P}_{\lambda}} \otimes\E_{U\sim \mathbb{U}_d}\sbra*{ \mathtt{q}_{\lambda}(U) \mathtt{q}_\lambda(D_{r,d}) \mathtt{q}_\lambda(U)^\dag} \label{eq-241815}\\
&=\bigoplus_{\lambda \vdash n} I_{\mathcal{P}_\lambda} \otimes I_{\mathcal{Q}_\lambda^d} \cdot \frac{s_{\lambda}(D_{r,d})}{\dim(\mathcal{Q}_\lambda^d)},\label{eq-220432}
\end{align}
where \cref{eq-241815} can be seen by \cref{fact-241853}, in \cref{eq-220432} is by \cref{coro-240446} and the observation that \(\E_{U\sim \mathbb{U}_d}\sbra*{ \mathtt{q}_{\lambda}(U)\mathtt{q}_{\lambda}(D_{r,d})\mathtt{q}_\lambda(U)^\dag}\) commutes with the actions of \(U\in\mathbb{U}_d\), in which \(s_{\lambda}(D_{r,d})\) refers to \(s_{\lambda}(\underbrace{1/r,\ldots,1/r}_r,\underbrace{0,\ldots,0}_{d-r})\).
Similarly, for the second case (i.e., \(\rho=\rho_2\)), we have
\begin{align}
\E\sbra*{\rho_2^{\otimes n}}=\bigoplus_{\lambda \vdash n} I_{\mathcal{P}_\lambda} \otimes I_{\mathcal{Q}_\lambda^d} \cdot \frac{s_{\lambda}(D_{d,d})}{\dim(\mathcal{Q}_\lambda^d)}.\nonumber
\end{align}

By \cref{thm:qsd}, the success probability of distinguishing \(\E[\rho_1^{\otimes n}]\) and \(\E[\rho_2^{\otimes n}]\) is upper bounded by 
\[\frac{1}{2}+\frac{1}{4}\Abs*{\E\sbra*{\rho_1^{\otimes n}}-\E\sbra*{\rho_2^{\otimes n}}}_{1}.\]
Note that
\begin{align}
\Abs*{\E\sbra*{\rho_1^{\otimes n}}-\E\sbra*{\rho_2^{\otimes n}}}_{1}&=\sum_{\lambda\vdash n} \left|\dim(\mathcal{P}_\lambda) \cdot s_\lambda(D_{r,d})-\dim(\mathcal{P}_\lambda) \cdot s_\lambda(D_{d,d})\right|\nonumber\\
&=\|\operatorname{SW}^n_r- \operatorname{SW}^n_d\|_1,\label{eq-242111}
\end{align}
where in \cref{eq-242111} we use the stability of Schur polynomial, i.e., \[s_{\lambda}(D_{r,d})=s_\lambda(\underbrace{\tfrac{1}{r},\ldots,\tfrac{1}{r}}_{r},\underbrace{0,\ldots,0}_{d-r})=s_\lambda(\underbrace{\tfrac{1}{r},\ldots,\tfrac{1}{r}}_{r}).\]
Then, since \(n\leq r \leq d \) (see \cref{eq-242207}), by \cref{lemma-241736}, we have that
\[\|\operatorname{SW}^n_r-\operatorname{Planch}(n)\|_1\leq \sqrt{2}\frac{n}{r},\]
and
\[\|\operatorname{SW}^n_d-\operatorname{Planch}(n)\|_1\leq \sqrt{2}\frac{n}{d}.\]
This means
\begin{align}
\|\operatorname{SW}^n_r- \operatorname{SW}^n_d\|_1 
&\leq \|\operatorname{SW}^n_r-\operatorname{Planch}(n)\|_1+\|\operatorname{SW}^n_d-\operatorname{Planch}(n)\|_1 \nonumber \\
&\leq \sqrt{2}\frac{n}{r}+\sqrt{2}\frac{n}{d}. \nonumber
\end{align}

Therefore, if the success probability is at least $2/3$, we have
\[\frac{2}{3} \leq \frac{1}{2} + \frac{1}{4}\rbra*{\sqrt{2}\frac{n}{r}+\sqrt{2}\frac{n}{d} } \leq \frac{1}{2} + \frac{n}{\sqrt{2}r},\]
which means
\begin{equation*}
n\geq \frac{\sqrt{2}r}{6} = \frac{\sqrt{2}}{6}\floor*{ \frac{1}{(2\varepsilon)^{\frac{1}{q-1}}}} =\Omega\rbra*{\frac{1}{\varepsilon^{\frac{1}{q-1}}}}.
\end{equation*}
\end{proof}

\section{The Case of \texorpdfstring{$0 < q < 1$}{0 < q < 1}}

\subsection{Upper bound}\label{sec-2100432}
\begin{algorithm}[h]
\caption{$\texttt{PowerTrace}\rbra{\rho, q, \varepsilon}$ for $0 < q < 1$}
    \label{algo-2100311}
    \begin{algorithmic}[1]
    \Require Sample access to a $d$-dimensional mixed quantum state $\rho$; $q \in \rbra{0, 1}$ and $\varepsilon \in \rbra{0, 1}$. 

    \Ensure An estimate of $\tr\rbra{\rho^q}$.
    
    \State $\eta\gets \varepsilon^{1/q}/d^{1/q-1}$, $n\gets \Theta(d^2/\eta^2)$.

    \State $\hat\alpha \gets \texttt{SpectrumEstimation}\rbra{\rho, n, 1}$. 

    \State $\hat P \gets \sum_{j=1}^{d} \hat\alpha_j^q$.

    \State \Return $\hat P$.
    \end{algorithmic}
\end{algorithm}

We state the sample complexity of estimating $\tr\rbra{\rho^q}$ for the case of $0 < q < 1$ as follows.

\begin{theorem} \label{thm-2100931}
For every constant $0<q<1$, we can estimate $\tr(\rho^q)$ to within additive error $\varepsilon$ using $O(d^{2/q}/\varepsilon^{2/q})$ samples of $\rho$.
\end{theorem}
\begin{proof}
Let $\eta\coloneqq \varepsilon^{1/q}/d^{1/q-1}$.
By \cref{lemma:tv-estimation}, we can use $n\coloneqq \Theta(d^2/\eta^2)=O((d/\varepsilon)^{2/q})$ samples of $\rho$ to obtain a sequence $\hat{\alpha}=(\hat{\alpha}_1,\hat{\alpha}_2,\ldots,\hat{\alpha}_d)$ such that
\[\sum_{j=1}^d |\hat{\alpha}_j-\alpha_j|\leq \eta,\]
with probability at least $2/3$.
Then, we consider the estimator
\[\hat{P}\coloneqq \sum_{j=1}^d \hat{\alpha}_j^q.\]
We have
\begin{align}
\left|\hat{P}-\tr(\rho^q)\right|&\leq 
\sum_{j=1}^d\abs*{\hat{\alpha}_j^q-\alpha_j^q} \nonumber  \\
&\leq \sum_{j=1}^d \abs*{\hat{\alpha}_j-\alpha_j}^q  \label{eq-2100401}\\
&\leq \left(\sum_{j=1}^d \abs*{\hat{\alpha}_j-\alpha_j}\right)^q d^{1-q} \label{eq-2100400}\\
& \leq \eta^q d^{1-q} =\varepsilon,\nonumber 
\end{align}
where \cref{eq-2100401} is due to \cref{prop-12251837}, and \cref{eq-2100400} is due to \cref{prop-12251915}.
Therefore, with probability at least $2/3$, $\hat{P}$ gives a good estimate for $\tr(\rho^q)$.
\end{proof}

\subsection{Lower bound}
\begin{theorem} \label{thm-03031025}
For every constant $0<q<1$ and sufficiently large $d$, any quantum estimator to additive error $\varepsilon$ for $\tr(\rho^q)$ requires sample complexity $\Omega(d^{1/q}/\varepsilon^{1/q})$.
\end{theorem}
Before giving the proof, we first introduce the following notations.
Let $n, m, d$ be positive integers such that $n\geq m$. Let $\mathcal{H}_1\cong\cdots\cong \mathcal{H}_n\cong\mathbb{C}^{d}$ be $n$ copies of the $d$-dimensional Hilbert space. 
Let $S\subseteq [n]=\{1,2,\ldots,n\}$ be a set of integers and $A$ be a $d\times d$ matrix. We use the following notation $A^{\otimes S}$
to denote the matrix $A^{\otimes |S|}$ acting on $\bigotimes_{i\in S} \mathcal{H}_i$.
Therefore, if $B$ is another $d\times d$ matrix, then
$A^{\otimes S} \otimes B^{\otimes [n]\setminus S}$
denotes the matrix $\bigotimes_{i=1}^n X_i$ acting on $\bigotimes_{i=1}^n \mathcal{H}_i$ where $X_i=A$ for $i\in S$, and $X_i=B$ otherwise.
Now we give the proof below.
\begin{proof}[Proof of \cref{thm-03031025}]
Given an integer \(r\in [1,d-1]\), and $\Delta\in [0,1]$, we use \(D_{r}\) to denote the \((d-1)\times (d-1)\) diagonal matrix:
\[D_{r}\coloneqq \diag(\underbrace{\tfrac{1}{r},\ldots,\tfrac{1}{r} }_{r},\underbrace{0,\ldots,0}_{d-1-r}),\]
with respect to the orthonormal states $\{\ket{1},\ldots,\ket{d-1}\}$ (i.e., $D_{r}=\frac{1}{r}\sum_{i=1}^r\ketbra{i}{i}$).
Let 
\begin{equation*}
r_1\coloneqq \left\lfloor \frac{d-1}{2}\right\rfloor,\quad r_2\coloneqq 2r_1, \quad \textup{and}\quad \Delta\coloneqq \frac{\varepsilon^{1/q}}{d^{1/q-1}}\cdot \frac{1}{(1-\frac{1}{2^{1-q}})^{1/q}}\leq \frac{\varepsilon^{1/q}}{d^{1/q-1}} \cdot O\!\left(\frac{1}{1-q}\right).
\end{equation*}

Then, consider the following distinguishing problem.
\begin{problem}\label{prob-2170137}
Suppose a \(d\)-dimensional quantum state \(\rho\) is in one of the  following with equal probability:
\begin{enumerate}
    \item \(\rho=\rho_1\coloneqq (1-\Delta)\cdot \ketbra{0}{0} + \Delta \cdot U D_{r_1} U^\dag\),
    \item \(\rho=\rho_2\coloneqq (1-\Delta) \cdot \ketbra{0}{0}+\Delta \cdot UD_{r_2}U^\dag\),
\end{enumerate}
where \(U\sim \mathbb{U}_{d-1}\) is a \(d-1\)-dimensional Haar random unitary acting on $\spanspace\{\ket{1},\ldots,\ket{d-1}\}$.
The task is to distinguish between the above two cases.
\end{problem}

Suppose an algorithm can estimate $\tr(\rho^q)$ to within error $\varepsilon/10$, with sample complexity $n$. Then, this algorithm can distinguish these two cases, since 
\begin{align}
\tr(\rho_2^q)-\tr(\rho_1^q)&=\Delta^q r_2^{1-q}-\Delta^q r_1^{1-q} \nonumber\\
&=\frac{\varepsilon}{ d^{1-q}(1-\frac{1}{2^{1-q}})}(r_2^{1-q}-r_1^{1-q}) \nonumber \\
&= \frac{\varepsilon r_2^{1-q}}{d^{1-q}} \nonumber \\
&\geq \frac{\varepsilon}{2}.\nonumber
\end{align}

However, to distinguish these two states, it must satisfy that:
\[n\geq \Omega(d/\Delta)= \Omega\left((1-q)\frac{d^{1/q}}{\varepsilon^{1/q}}\right).\]
To see this, we first assume that $10000n\Delta \leq d$, since otherwise $n\geq \Omega(d/\Delta)$ and we have done.
Note that
\begin{align}
\E_{U\sim\mathbb{U}_{d-1}}[\rho_1^{\otimes n}] &= \E_{U\sim\mathbb{U}_{d-1}}\!\left[\sum_{i=0}^n \Delta^{i} (1-\Delta)^{n-i} \cdot \sum_{\substack{S\subseteq [n]\\ |S|=i}}\left(UD_{r_1}U^\dag\right)^{\otimes S} \otimes \ketbra{0}{0}^{\otimes [n]\setminus S}\right] \nonumber \\
&=\sum_{i=0}^n \Delta^{i} (1-\Delta)^{n-i} \cdot \sum_{\substack{S\subseteq [n]\\ |S|=i}}\overline{D}_{r_1}^S \otimes \ketbra{0}{0}^{\otimes [n]\setminus S}, \nonumber
\end{align}
where $\overline{D}_{r}^S\coloneqq \E_{U\sim\mathbb{U}_{d-1}}\!\left[(UD_{r}U^\dag)^{\otimes S}\right]$. Similarly,
\[\E_{U\sim\mathbb{U}_{d-1}}[\rho_2^{\otimes n}]=\sum_{i=0}^n \Delta^{i} (1-\Delta)^{n-i} \cdot \sum_{\substack{S\subseteq [n]\\ |S|=i}}\overline{D}_{r_2}^S \otimes \ketbra{0}{0}^{\otimes [n]\setminus S}.\]
Then, we can see that
\begin{align}
\left\|\E_U[\rho_1^{\otimes n}]-\E_U[\rho_2^{\otimes n}]\right\|_1 &\leq \sum_{i=0}^n\Delta^i(1-\Delta)^{n-i}\cdot \sum_{\substack{S\subseteq[n]\\ |S|=i}}\left\|\overline{D}_{r_1}^S-\overline{D}_{r_2}^S\right\|_1 \nonumber\\
&=\sum_{i=1}^n \Delta^i(1-\Delta)^{n-i}\cdot \binom{n}{i}\cdot  \left\|\overline{D}_{r_1}^i-\overline{D}_{r_2}^i\right\|_1 \label{eq-322040} \\
&\leq \sum_{1\leq i\leq \min\{n,100 n \Delta\}} \Delta^i (1-\Delta)^{n-i}\cdot \binom{n}{i} \cdot \left\|\overline{D}_{r_1}^i-\overline{D}_{r_2}^i\right\|_1 + \frac{e}{100}\label{eq-322111}\\
&=\sum_{1\leq i\leq \min\{n,100n\Delta\}} \Delta^i (1-\Delta)^{n-i}\cdot \binom{n}{i} \cdot \left\|\operatorname{SW}^i_{r_1}-\operatorname{SW}^i_{r_2}\right\|_1 +\frac{e}{100} \label{eq-322203} \\
&\leq \sum_{1\leq i\leq \min\{n,100n\Delta\}} \Delta^i (1-\Delta)^{n-i}\cdot \binom{n}{i} \cdot \left(\sqrt{2}\frac{i}{r_1}+\sqrt{2}\frac{i}{r_2}\right) +\frac{e}{100} \label{eq-322202} \\
&\leq \frac{1000n\Delta}{d}+\frac{e}{100}, \label{eq-322201}
\end{align}
where in \cref{eq-322040} we define $\overline{D}_r^i=\E_{U\sim\mathbb{U}_{d-1}}\!\left[(UD_rU^\dag)^{\otimes i}\right]$, \cref{eq-322111} is due to \cref{lemma-322325}, \cref{eq-322203} is by noting that (similar to \cref{eq-241815} and \cref{eq-220432})
\[\overline{D}_{r}^i=\E_{U\sim\mathbb{U}_{d-1}}\!\left[U^{\otimes i} D_{r}^{\otimes i} U^{\dag\otimes i}\right]=\bigoplus_{\lambda\vdash n} I_{\mathcal{P}_\lambda}\otimes I_{\mathcal{Q}_\lambda^{d-1}}\cdot \frac{s_\lambda(D_r)}{\dim(\mathcal{Q}_\lambda^{d-1})},\]
\cref{eq-322202} is because
\[\left\|\operatorname{SW}_{r_1}^i-\operatorname{SW}_{r_2}^i\right\|\leq \left\|\operatorname{SW}_{r_1}^i-\operatorname{Planch}(i)\right\|+\left\|\operatorname{SW}_{r_2}^i-\operatorname{Planch}(i)\right\|\leq \sqrt{2}\frac{i}{r_1}+\sqrt{2}\frac{i}{r_2},\]
due to \cref{lemma-241736} and the fact that $i\leq 100n\Delta\leq d/100<r_1<r_2$, and \cref{eq-322201} is because $\sum_{1\leq i\leq \min\{n,100n\Delta\}}\Delta^i(1-\Delta)^{n-i}\binom{n}{i}\leq 1$.
Therefore, if the success probability is at least $2/3$, then due to \cref{thm:qsd}, we have
\[\frac{2}{3}\leq \frac{1}{2}+\frac{1}{4}\left(\frac{1000n\Delta}{d}+\frac{e}{100}\right),\]
which means
\[n\geq \Omega(d/\Delta),\]
as desired.
\end{proof}
\begin{lemma}\label{lemma-322325}
Let $n>0$ be an integer and $p\in [0,1]$. Then, for any $k>e$, we have
\[\sum_{knp< i\leq n}p^i(1-p)^{n-i}\binom{n}{i}\leq \frac{e}{k}.\]
\end{lemma}
\begin{proof}
Let $X_1,\ldots,X_n$ be independent Bernoulli random variables such that $\Pr[X_i=1]=p$. Let $X=\sum_{i=1}^n X_i$ and note that $\E[X]=np$. If $p=0$, then the claim is trivial. Hence, we may assume $p>0$.
Let $m=\max\{1,knp\}$ and $\gamma=m/np$. Then, we have
\begin{align}
\sum_{knp<i\leq n}p^i(1-p)^{n-i}\binom{n}{i}&\leq \sum_{m\leq i\leq n}p^i(1-p)^{n-i}\binom{n}{i} \leq \Pr[X\geq m]\nonumber \\
&\stackrel{\hypertarget{eq-330047}{(\textup{i})}}{\leq} \left(\frac{e^{\gamma-1}}{\gamma^\gamma}\right)^{np}\nonumber \leq \left(\frac{e}{\gamma}\right)^{\gamma np} \nonumber \\
&\stackrel{\hypertarget{eq-330048}{(\textup{ii})}}{\leq} \frac{e}{k},\nonumber 
\end{align}
where \hyperlink{eq-330047}{(i)} is by the Chernoff's bound, and \hyperlink{eq-330048}{(ii)} is because $\gamma\geq k>e$ and $\gamma np \geq 1$.
\end{proof}

\subsection{Computational Hardness}
\label{subsec:comp-hardness-0leQle1}

We begin with the promise problems corresponding to estimating the quantum $q$-Tsallis entropy: 

\begin{definition}[Quantum $q$-Tsallis Entropy Approximation, \TsallisQEA{}, adapted from~{\cite[Definition 5.2]{LW25}}]
	\label{def:TsallisQEA}
    Let $Q$ be a quantum circuit acting on $m$ qubits with $n$ specified output qubits, where $m\geq 2n$. Let $\rho$ denote the quantum state obtained by applying $Q$ to the initial state $\ket{0}^{\otimes m}$ and tracing out the non-output qubits. Let $g(n)$ and $t(n)$ be positive efficiently computable functions. The task is to decide whether the following holds:
    \begin{itemize}
	   \item \textbf{\emph{Yes:}} $\Sq(\rho) \geq t(n) + g(n)$;
	   \item \textbf{\emph{No:}} $\Sq(\rho) \leq t(n) - g(n)$.
    \end{itemize}
\end{definition}

The main result of this subsection is the following hardness statement, which directly follows from a reduction from \QSCMM{} to \TsallisQEA{} for   $q\in(0,1)$, as detailed in \Cref{lemma:reduction-MaxMixedQSD-TsallisQEA}: 

\begin{restatable}[$\TsallisQEA$ is \NIQSZK{}-hard for $0< q <1$]{theorem}{TsallisQEAhardness}
    \label{thm:TsallisQEA-NIQSZKhard}
    For all real $q\in(0,1)$ and all integers $n \geq \ceil[\big]{\frac{5}{q(1-q)}}$, the following holds\emph{:}
    \[ \forall g(n;q) \in \sbra[\bigg]{\frac{1}{\poly(n)}, \frac{2^{5/q}q}{40}},~ \TsallisQEA[t(n;q),g(n;q)] \text{ is } \NIQSZK{}\text{-hard}. \]
    The threshold parameter $t(n;q)$ is as defined in \Cref{eq:tn-reduction-MaxMixedQSD-TsallisQEA}. 
\end{restatable}

To establish the reduction $\QSCMM \leq \TsallisQEA$, we prove an inequality that relates the $q$-Tsallis entropy of a distribution to its total variation distance from the uniform distribution. Specifically, combining \Cref{lemma:inequality-uniformTV-TsallisEA-LB,lemma:inequality-uniformTV-TsallisEA-UB} yields the desired bounds:\footnote{Related results have previously been shown: for the case $q=1$, see in~\cite[Lemma 16]{KLGN19}; and for the case $q>1$ with the additional condition $\TV(p,\nu) \geq 1/q$ (tailored to the upper bound), see~\cite[Lemma 4.10]{LW25}.}
\begin{theorem}[Bounds on $q$-Tsallis entropy via closeness to uniform distribution for $0<q<1$]
    \label{thm:inequality-uniformTV-TsallisEA}
    Let $p$ be a probability distribution over $[N]$, where $N\geq 2$, and let $\nu$ denote the uniform distribution over $[N]$. 
    Then, for every $q\in(0,1)$, the following inequalities hold\emph{:}
    \[ \rbra*{ 1-\TV(p,\nu) - \frac{1}{N} } \ln_q(N) \leq \Hq(p) \leq \ln_q(N) - \frac{q}{2}N^{1-q} \cdot \TV(p,\nu)^2. \]
\end{theorem}

In the remainder of this subsection, we present the proofs of \Cref{thm:TsallisQEA-NIQSZKhard,thm:inequality-uniformTV-TsallisEA}. Notably, for $q=1/2$, the upper bound in \Cref{thm:inequality-uniformTV-TsallisEA} can be strengthened using a simpler argument (\Cref{lemma:inequality-uniformTV-TsallisEA-UBhalf}). Consequently, this yields a slightly simpler reduction $\QSCMM \leq \TsallisQEAnoq_{1/2}$ (\Cref{lemma:reduction-MaxMixedQSD-TsallisQEA-half}) and the corresponding hardness result (\Cref{thm:TsallisQEA-NIQSZKhard-half}). 

\subsubsection{Useful bounds on Tsallis entropy}
\label{subsubsec:Tsallis-entropy-bounds}

In this subsection, we establish bounds on the $q$-Tsallis entropy, as stated in \Cref{lemma:inequality-uniformTV-TsallisEA-LB,lemma:inequality-uniformTV-TsallisEA-UB,lemma:inequality-uniformTV-TsallisEA-UBhalf}.

\begin{lemma}[Lower bound on $q$-Tsallis entropy via closeness to uniform distribution for $0<q<1$]
    \label{lemma:inequality-uniformTV-TsallisEA-LB}
    Let $p$ be a probability distribution over $[N]$, where $N\geq 2$, and let $\nu$ denote the uniform distribution over $[N]$. 
    Then, for every $q\in(0,1)$, the following inequality holds\emph{:}
    \[\rbra*{ 1-\TV(p,\nu) - \frac{1}{N} } \ln_q(N) \leq \Hq(p).\]
\end{lemma}

\begin{proof}
    Let $\Delta_{N}$ denote the set of probability distributions over $N$ elements. For any distribution $p\in \Delta_{N}$, it is straightforward to verify that $0 \leq \TV(p,\nu) \leq 1-1/N$. Throughout, we use the shorthand $\gamma \coloneqq \TV(p,\nu)$ for convenience. 
    
    To demonstrate the lower bound, we recall that $\Hq(p)$ is concave for all $0<q<1$ (see, e.g.,~\cite[Lemma 2.3]{LW25}). Consequently, the proof strategy used for the cases $q=1$ and $q>1$, particularly those in~\cite[Lemma 16]{KLGN19} and~\cite[Lemma 4.10]{LW25}, applies in the present setting as well. Specifically, it suffices to solve the convex optimization problem in \Cref{eq:Tsallis-entropy-min}, whose objective is to minimize the Tsallis entropy $\Hq(p)$ under the constraint $\TV(p,\nu) = \gamma$. Since the feasible region specified in \Cref{eq:Tsallis-entropy-min} constitutes a closed convex set, the problem admits an optimal solution, one such minimizer $p_{\min}$ is given in \Cref{eq:Tsallis-entropy-min-solution}.\footnote{Since $\Hq(p)$ is concave, its minimum over the compact convex polytope in \Cref{eq:Tsallis-entropy-min} is attained at an extreme point. A direct check of extremality shows that, up to permutation, the extreme points are exactly those of the form given in \Cref{eq:Tsallis-entropy-min-solution}.} 
    
    {\centering
    \begin{minipage}{0.5\textwidth}
        \centering
        \begin{mini}{}{\Hq(p')}{}{}{\label{eq:Tsallis-entropy-min}}{}
            \addConstraint{p'}{\in\Delta_N}{}
            \addConstraint{\TV(p',\nu)}{\leq \gamma}{}
        \end{mini}
    \end{minipage}%
    \begin{minipage}{0.5\textwidth}
    \begin{subequations}
    \label{eq:Tsallis-entropy-min-solution}
    \begin{align}
        &p_{\min}(i) = \begin{cases}
            \frac{1}{N}, &\text{if } i\in[k_{\min}]\\
            \frac{1}{N}+\gamma, &\text{if } i=k_{\min}+1\\
            \frac{\varepsilon}{N}, &\text{if } i=k_{\min}+2\\
            0, &\text{otherwise}
        \end{cases},\\
        &\begin{aligned}
            &\text{where } &k_{\min} &\coloneqq \floor*{N(1-\gamma)}-1,\\
            &&\varepsilon &\coloneqq N(1\!-\!\gamma) - \floor*{N(1\!-\!\gamma)}.
        \end{aligned}
    \end{align}
    \end{subequations}
    \end{minipage}  
    \vspace{1em}
    }

    Next, we derive the lower bound for the Tsallis entropy by evaluating $\Hq(p_{\min})$: 
    \begin{align*}
        \Hq(p_{\min}) &= \frac{1}{1-q} \left( k_{\min} \Big( \frac{1}{N} \Big)^q + \Big(\frac{1}{N} + \gamma\Big)^q + \Big( \frac{\varepsilon}{N} \Big)^q - 1 \right)\\
        &\geq \frac{1}{1-q} \left( (\floor*{N(1-\gamma)}-1 + \varepsilon^q) N^{-q} - \frac{(\floor*{N(1-\gamma)}-1)}{N} - \frac{\varepsilon}{N}\right)\\ 
        &= \frac{1}{1-q} \left( ( N(1-\gamma)-\varepsilon -1 + \varepsilon^q) N^{-q} - \frac{N(1-\gamma)-1}{N}\right)\\ 
        &\geq \frac{1}{1-q} \left( \Big(1-\gamma -\frac{1}{N}\Big) N^{1-q} - \rbra[\Big]{1-\gamma-\frac{1}{N}} \right)\\ 
        &= \left(1-\gamma - \frac{1}{N} \right) \ln_q(N).
    \end{align*}
    Here, in the second line, we omit the terms corresponding to $p_{\min}(k_{\min}+1)$, since $\rbra[\big]{\frac{1}{N}+\gamma}^q - \rbra[\big]{\frac{1}{N}+\gamma} \geq 0$ for $0 < q < 1$. The third line follows from the identity $\floor*{N(1-\gamma)} = N(1-\gamma) - \varepsilon$. The fourth line holds because $\varepsilon^q - \varepsilon \geq 0$ for $0 < q < 1$ and $0\leq \varepsilon \leq 1$.
\end{proof}

\begin{lemma}[Upper bound on $q$-Tsallis entropy via closeness to uniform distribution for $0<q<1$] 
    \label{lemma:inequality-uniformTV-TsallisEA-UB}
    Let $p$ be a probability distribution over $[N]$, where $N\geq 2$, and let $\nu$ denote the uniform distribution over $[N]$. 
    Then, for every $q\in(0,1)$, the following inequality holds\emph{:}
    \[ \Hq(p) \leq \ln_q(N) - \frac{q}{2}N^{1-q} \cdot \TV(p,\nu)^2. \]
\end{lemma}

\begin{proof}
    As in \Cref{lemma:inequality-uniformTV-TsallisEA-LB}, for any distribution $p\in \Delta_{N}$, we have $0 \leq \gamma \coloneqq\TV(p,\nu) \leq 1-1/N$. 
    To prove the upper bound, it therefore suffices to maximize the Tsallis entropy $\Hq(p)$ under the constraint $\TV(p,\nu) = \gamma$, which can be formulated as a \textit{non-convex} optimization problem analogous to \Cref{eq:Tsallis-entropy-min}. Accordingly, we consider the following optimization problem: 
    
    \begin{maxi}{}{\Hq(p')}{}{}{\label{eq:Tsallis-entropy-max}}{}
        \addConstraint{p'}{\in\Delta_N}{}
        \addConstraint{\TV(p',\nu)}{\geq \gamma.}{}
    \end{maxi}

    Although \Cref{eq:Tsallis-entropy-max} is analyzed in~\cite[Lemma 4.10]{LW25} for the regime $q>1$, the case $0<q<1$ requires a different analysis. 
    Since $\Hq(p)$ is maximized at $\nu$, any optimizer of \Cref{eq:Tsallis-entropy-max} may be taken to satisfy $\TV(p_{\max},\nu)=\gamma$, and by Jensen's inequality, the optimizer must be equalized on each side of $1/N$.
    Hence, any optimal solution $p_{\max}$ must have the form:\footnote{To see this, let $p'\in\Delta_N$ be feasible for \Cref{eq:Tsallis-entropy-max}, and write $\gamma' \coloneqq \TV(p',\nu) \geq \gamma$. If $\gamma'>\gamma$, then for $\tilde{p} \coloneqq \frac{\gamma}{\gamma'} p'+ \rbra[\big]{1-\frac{\gamma}{\gamma'}}\nu$, we have $\TV(\tilde{p},\nu)=\gamma$, and by concavity of $\Hq(p)$ together with the fact that $\Hq(\nu)$ is maximal, $\Hq(\tilde{p}) \geq \Hq(p')$. Hence it suffices to consider the case $\TV(p',\nu)=\gamma$. Now let $A\coloneqq \cbra{i:p'(i)>1/N}$ and $k\coloneqq |A|$. Then $\sum_{i\in A}\rbra*{p'(i)-1/N} = \gamma$ and $\sum_{i\notin A} \rbra*{1/N-p'(i)} = \gamma$. Since $x \mapsto x^q$ is concave for $0<q<1$, Jensen's inequality gives $\sum_{i\in A} p'(i)^q \leq k\rbra[\big]{\frac{1}{N} + \frac{\gamma}{k}}^q$ and $\sum_{i\notin A} p'(i)^q \leq (N-k) \rbra[\big]{\frac{1}{N}-\frac{\gamma}{N-k}}^q$. Therefore, $\sum_i \psi_i p'(i)^q \leq \PS_q(N,k,\gamma)$, so $\Hq(p') \leq \Hq\rbra[\big]{p^{(k)}}$. Thus, an optimal solution is attained by some distribution of the form $p^{(k)}$ in \Cref{eq:Tsallis-max-solution-candiates}, where $k \leq N(1-\gamma)$ ensures entry-wise non-negativity.\label{footnote:opt-solution-form}}
    \begin{equation}
        \label{eq:Tsallis-max-solution-candiates}
        \Hq(p_{\max}) = \max_{k \in [\floor*{N(1-\gamma)}]} \Hq\rbra*{p^{(k)}},
        \text{ where }
        p^{(k)}(i) \coloneqq \begin{cases}
        \frac{1}{N}+\frac{\gamma}{k}, &\text{if } i\in[k],\\
        \frac{1}{N} - \frac{\gamma}{N-k}, &\text{otherwise}.
        \end{cases}
    \end{equation}

    Combining \Cref{eq:Tsallis-entropy-max,eq:Tsallis-max-solution-candiates}, the task reduces to the following optimization problem for $0<q<1$, where the optimal choice of $k$ depends on the order $q$:
    \begin{maxi}{k}{\PS_q(N,k,\gamma) \coloneqq \sum_{i\in[N]} p^{(k)}(i)^q = k \cdot \rbra*{\frac{1}{N} + \frac{\gamma}{k}}^q + (N-k) \cdot \rbra*{\frac{1}{N} - \frac{\gamma}{N-k}}^q}{}{}{\label{eq:Tsallis-entropy-max-powerSum}}{}
        \addConstraint{0}{ \leq \gamma \leq 1-1/N}{}
        \addConstraint{1}{\leq k \leq \floor*{N(1-\gamma)}}{}
        \addConstraint{k, N}{\in \bbZ_+}{}
    \end{maxi}

    In the remainder of the proof, we derive an upper bound on the optimal value of \Cref{eq:Tsallis-entropy-max-powerSum}. 
    To this end, set $t \coloneqq k/N \in [1/N,1-\gamma]$. A straightforward calculation shows that
    \begin{align*}
        N^{q-1} \cdot \PS_q(N,k,\gamma) &= \rbra*{\frac{k}{N}}^{1-q} \rbra*{\frac{k}{N}+\gamma}^q + \rbra*{\frac{N-k}{N}}^{1-q} \rbra*{\frac{N-k}{N}-\gamma}^q\\
        &= t^{1-q} (t+\gamma)^q + (1-t)^{1-q} (1-t-\gamma)^q \coloneq F_q(t,\gamma).
    \end{align*}

    Applying Taylor's theorem to $F_q(t,\gamma)$ around $\gamma=0$, there exists $\xi \in (0,\gamma)$ corresponding to the Lagrange remainder such that 
    \begin{equation}
        \label{eq:truncationUB}
        F_q(t,\gamma) = F_q(t,0) + \gamma \cdot \frac{\partial}{\partial \gamma} F_q(t,\gamma)\bigg|_{\gamma=0} + \frac{\gamma^2}{2} \cdot \frac{\partial^2}{\partial \gamma^2} F_q(t,\gamma)\bigg|_{\gamma=\xi}.
    \end{equation}
    The first and second derivatives are given by
    \begin{subequations}
    \label{eq:derivativesUB}
    \begin{align}
        \frac{\partial}{\partial\gamma} F_q(t,\gamma) &= q t^{1-q} (t+\gamma)^{q-1} - q (1-t)^{1-q} (1-t-\gamma)^{q-1},\\
        \frac{\partial^2}{\partial\gamma^2} F_q(t,\gamma) &= q(q-1)\rbra*{ t^{1-q} (t+\gamma)^{q-2} + (1-t)^{1-q} (1-t-\gamma)^{q-2} }.
    \end{align}
    \end{subequations}
    Combining \Cref{eq:truncationUB,eq:derivativesUB}, to establish the desired inequality it suffices to prove that    
    \begin{equation}
        \label{eq:targetInequalityUB}
        1 + \frac{\gamma^2}{2} \cdot \frac{\partial^2}{\partial \gamma^2} F_q(t,\gamma) \bigg|_{\gamma=\xi} \leq 1 - \frac{q}{2}(1-q)\gamma^2. 
    \end{equation}

    We now establish \Cref{eq:targetInequalityUB} by deriving the following lower bound via direct calculation: 
    \begin{subequations}
    \label{eq:targetInequalityUBfinal}
    \begin{align}
        \frac{1}{q(q-1)} \cdot \frac{\partial^2}{\partial \gamma^2} F_q(t,\gamma) \bigg|_{\gamma=\xi} &= t^{1-q} (t+\xi)^{q-2} + (1-t)^{1-q} (1-t-\xi)^{q-2}\\
        &= \frac{1}{t} \cdot \rbra*{1+\frac{\xi}{t}}^{q-2} + \frac{1}{1-t}\cdot \rbra*{1- \frac{\xi}{1-t}}^{q-2}\\
        &\geq \frac{1}{t} \cdot \rbra*{1+\frac{1}{t}}^{-2} + \frac{1}{1-t} = \frac{3t+1}{(1-t)(t+1)^2} \coloneqq G(t).
    \end{align}
    \end{subequations}
    Here, the inequality in the third line follows from 
    \[  \rbra*{1+\frac{\xi}{t}}^{q-2} \geq \rbra*{1+\frac{1}{t}}^{q-2} \geq \rbra*{1+\frac{1}{t}}^{-2} \quad\text{and}\quad \rbra*{1-\frac{\xi}{1-t}}^{q-2} \geq 1^{q-2} = 1, \]
    which are consequences of the monotonicity of $f(x)=x^{q-2}$ for all $0<q<1$.

    Finally, \Cref{eq:targetInequalityUB} follows from the bound in \Cref{eq:targetInequalityUBfinal} because $G(t)\geq 1$, which holds since $(3t+1)-(1-t)(t+1)^2 = t(t^2+t+2) \geq 0$ for all $t \geq 0$, and thus the proof is complete. 
\end{proof}

\begin{lemma}[Upper bound on $q$-Tsallis entropy via closeness to uniform distribution for $q=1/2$]
    \label{lemma:inequality-uniformTV-TsallisEA-UBhalf}
    Let $p$ be a probability distribution over $[N]$, where $N\geq 2$, and let $\nu$ denote the uniform distribution over $[N]$. 
    Then, for $q=1/2$, the following inequality holds\emph{:}
    \[ \Hq(p) \leq \ln_q\rbra*{ N (1-\TV(p,\nu)^2) }.\]
\end{lemma}

\begin{proof}
    We adopt the same setting as in the proof of \Cref{lemma:inequality-uniformTV-TsallisEA-UB}. To establish the desired upper bound, it therefore suffices to consider the optimization problem in \Cref{eq:Tsallis-entropy-max}, which, by the same argument as in \Cref{footnote:opt-solution-form}, admits an optimal solution of the form given in \Cref{eq:Tsallis-max-solution-candiates}. Analogously, it remains to bound the objective function, subject to the constraints in \Cref{eq:Tsallis-entropy-max-powerSum}, which for $q=1/2$ becomes: 
    \begin{align*}
        \PS(N,k,\gamma) \coloneqq \PS_{1/2}(N,k,\gamma)
        &= k \cdot \sqrt{\frac{1}{N} + \frac{\gamma}{k}} + (N-k) \cdot \sqrt{\frac{1}{N} - \frac{\gamma}{N-k}}\\
        &= \frac{1}{\sqrt{N}} \rbra[\Big]{ \underbrace{\sqrt{k(k+\gamma N)}}_{F_0(k;\gamma,N)} + \underbrace{\sqrt{ (N-k) ( (1-\gamma) N - k) }}_{F_1(k;\gamma,N)} }.
    \end{align*}

    Define $k_{\pm} \coloneqq \frac{N(1-\gamma)}{2} \pm t$, so that $k_+ +k_- = N(1-\gamma)$. It is readily verified that 
    \[F_0(k_\pm;\gamma,N) = \sqrt{ \frac{N(1+\gamma)}{2} \pm t } \cdot\sqrt{ \frac{N(1-\gamma)}{2} \pm t } = F_1(k_\mp;\gamma,N).\] 
    
    Consequently, $\PS(N,k,\gamma)$ is symmetric about $k = N(1-\gamma)/2$. Moreover, it is straightforward to verify that $\PS(N,k,\gamma)$ is concave in $k$, since the following holds:
    \begin{align*}
        \frac{\partial^2}{\partial k^2} F_0(k;\gamma,N) &= -\frac{\gamma ^2 N^2}{4 (k (k+\gamma  N))^{3/2}} < 0,\\ 
        \frac{\partial^2}{\partial k^2} F_1(k;\gamma,N) &= -\frac{\gamma ^2 N^2}{4 ((N-k) ((1-\gamma) N - k))^{3/2}} < 0.
    \end{align*}

    Therefore, we conclude that the optimal solution of \Cref{eq:Tsallis-entropy-max-powerSum} occurs at either $k_*=\floor*{N(1-\gamma)/2}$ or $k_*=\ceil*{N(1-\gamma)/2}$. As a consequence, 
    \[ \Hq(p_{\max}) \leq 2 \rbra*{\PS\rbra*{N,\frac{N(1-\gamma)}{2},\gamma}-1} = 2\rbra*{\sqrt{N(1-\gamma^2)}-1} = \ln_{1/2}\rbra*{\sqrt{N(1-\gamma^2)}}. \qedhere\]
\end{proof}


\subsubsection{\NIQSZK{} hardness results}

As a warm-up, we first describe the reduction from \QSCMM{} to $\TsallisQEAnoq_{1/2}$:
\begin{lemma}[$\QSCMM\leq\TsallisQEAnoq_{1/2}$]
    \label{lemma:reduction-MaxMixedQSD-TsallisQEA-half}
    Let $\rho$ be an $n$-qubit quantum state whose purification can be prepared by an $m$-qubit quantum circuit $Q$, as described in \Cref{def:TsallisQEA}. For any such state $\rho$, define the threshold function 
    \begin{equation}
        \label{eq:tn-reduction-MaxMixedQSD-TsallisQEA-half}
        t(n) \coloneqq \rbra*{2^{n/2}-1}\rbra*{1-1/n-2^{-n}}+2^{n/2}\cdot\frac{\sqrt{2n-1}}{n}-1.
    \end{equation}
    Then, the following statements hold\emph{:}
    \[ 
        \forall n \geq 4,\quad 
        \begin{array}{ll}
        \td\rbra*{\rho,(I/2)^{\otimes n}} \leq 1/n &\;\Rightarrow\; \SqHalf(\rho) \geq t(n) + 1/10,\\
        \td\rbra*{\rho,(I/2)^{\otimes n}} \geq 1-1/n &\;\Rightarrow\;  \SqHalf(\rho) \leq t(n) - 1/10.
        \end{array}
    \]
\end{lemma}

\begin{proof}
    Given a quantum state $\rho$, let $\{v_i\}_{i\in[2^n]}$ be an orthonormal basis consisting of eigenvectors of $\rho$, so that it admits the spectral decomposition $\rho = \sum_{i\in[2^{n}]} \lambda_i \ketbra{v_i}{v_i}$. Let $p \coloneqq (\lambda_1, \cdots, \lambda_{2^n})$ denote the vector of eigenvalues, which forms a probability distribution of dimension $2^n$, satisfying $\SqHalf(\rho) = \HqHalf(p)$. Let $\nu$ denote the uniform distribution over $2^n$ elements.
    Since $\rho$ and $(I/2)^{\otimes n}$ commute, they are simultaneously diagonalizable in the same basis. Hence, 
    \[\td\rbra*{\rho, (I/2)^{\otimes n}} = \TV(p,\nu),\] 
    and the bounds in \Cref{lemma:inequality-uniformTV-TsallisEA-LB,lemma:inequality-uniformTV-TsallisEA-UBhalf} naturally extend to the quantum setting. 
    
    We now consider the following two cases: 
    \begin{itemize}
    \item In the case where $\td\rbra*{\rho,(I/2)^{\otimes n}} \leq 1/n$, the lower bound from \Cref{lemma:inequality-uniformTV-TsallisEA-LB} implies 
    \begin{align*}
        \SqHalf(\rho) &\geq \ln_{1/2}\big(2^n\big) \cdot \rbra*{1 - \td\rbra*{\rho,(I/2)^{\otimes n}} - 2^{-n}}
        \geq 2\rbra*{2^{n/2}-1} \rbra*{1 - \frac{1}{n} - 2^{-n}} \coloneqq \tau_{\ttY}(n).
    \end{align*}
    
    \item In the case where $\td\rbra*{\rho,(I/2)^{\otimes n}} \geq 1-1/n$, the upper bound from \Cref{lemma:inequality-uniformTV-TsallisEA-UBhalf} yields 
    \begin{align*}
        \SqHalf(\rho) &\leq \ln_{1/2}\rbra*{2^n\rbra*{1-\td\rbra*{\rho,(I/2)^{\otimes n}}^2}}
        \leq 2^{\frac{n}{2}+1}\frac{\sqrt{2n-1}}{n} -2
        \coloneqq \tau_\ttN(n).
    \end{align*}
    \end{itemize}

    We then define the threshold function $t(n) \coloneqq \frac{1}{2}\rbra*{ \tau_\ttY(n) + \tau_\ttN(n)}$ and the gap function 
    \[g(n) \coloneqq \frac{1}{2}\rbra*{ \tau_\ttY(n) - \tau_\ttN(n)} 
    = \underbrace{2^{n/2}}_{G_1(n)} \underbrace{\rbra[\Big]{1-\frac{1}{n}-2^{-n}-\frac{\sqrt{2n-1}}{n}}}_{G_2(n)}+\underbrace{\frac{1}{n}+2^{-n}}_{G_3(n)}.\] 
    
    It suffices to show that $g(n) \geq 1/5$ for all integers $n\geq 4$. 
    Since $G_3(n)$ is monotonically decreasing for $n\geq 4$ and satisfies $\lim_{n\rightarrow\infty} G_3(n)=0$, it follows that $G_3(n)\geq 0$ for $n\geq 1$, and thus this term does not contribute to the lower bound. Also, it is straightforward to verify that both $G_1(n)$ and $G_2(n)$ are positive and monotonically increasing for $n\geq 4$. Consequently, we complete the proof by noting that
    \[ g(n) = G_1(n) G_2(n)+G_3(n) \geq G_1(n) G_2(n) \geq G_1(4) G_2(4) = \frac{11}{4}-\sqrt{7} > \frac{1}{10}. \qedhere \]
\end{proof}

Next, we establish the reduction from \QSCMM{} to $\TsallisQEA$ for all $q\in(0,1)$:
\begin{lemma}[$\QSCMM\leq\TsallisQEA$ for $0<q<1$]
    \label{lemma:reduction-MaxMixedQSD-TsallisQEA}
    Let $\rho$ be an $n$-qubit quantum state whose purification can be prepared by an $m$-qubit quantum circuit $Q$, as described in \Cref{def:TsallisQEA}. For any such state $\rho$, define the threshold function 
    \begin{equation}
        \label{eq:tn-reduction-MaxMixedQSD-TsallisQEA}
        t(n;q) \coloneqq \frac{2^{n(1-q)}-1}{2(1-q)}\rbra*{2 - \frac{1}{n} - 2^{-n}} - \frac{q}{4}\cdot 2^{n(1-q)} \rbra*{1-\frac{1}{n}}^2.
    \end{equation}
    Then, the following statements hold\emph{:}
    \[ 
        \forall n \geq \ceil[\Big]{\frac{5}{q(1-q)}},\quad 
        \begin{array}{ll}
        \td\rbra*{\rho,(I/2)^{\otimes n}} \leq 1/n &\;\Rightarrow\; \Sq(\rho) \geq t(n;q) + 2^{5/q} q/40,\\
        \td\rbra*{\rho,(I/2)^{\otimes n}} \geq 1-1/n &\;\Rightarrow\;  \Sq(\rho) \leq t(n;q) - 2^{5/q} q/40.
        \end{array}
    \]
\end{lemma}

\begin{proof}
    Following the same reasoning at the beginning of the proof of \Cref{lemma:reduction-MaxMixedQSD-TsallisQEA-half}, the bounds in \Cref{lemma:inequality-uniformTV-TsallisEA-LB,lemma:inequality-uniformTV-TsallisEA-UB} also extend naturally to the quantum setting. 
    
    We now consider the following two cases: 
    \begin{itemize}
    \item In the case where $\td\rbra*{\rho,(I/2)^{\otimes n}} \leq 1/n$, the lower bound from \Cref{lemma:inequality-uniformTV-TsallisEA-LB} implies 
    \begin{align*}
        \Sq(\rho) &\geq \ln_{q}\!\big(2^n\big) \cdot \rbra*{1 - \td\rbra*{\rho,(I/2)^{\otimes n}} - 2^{-n}}
        \geq \frac{2^{n(1-q)}-1}{1-q}\rbra*{1 - \frac{1}{n} - 2^{-n}} \coloneqq \tau_{\ttY}(n;q).
    \end{align*}
    
    \item In the case where $\td\rbra*{\rho,(I/2)^{\otimes n}} \geq 1-1/n$, the upper bound from \Cref{lemma:inequality-uniformTV-TsallisEA-UB} yields 
    \begin{align*}
        \Sq(\rho) &\leq \ln_q(2^n) - \frac{q}{2} 2^{n(1-q)} \cdot \td\rbra*{\rho,(I/2)^{\otimes n}}^2        
        \leq \frac{2^{n(1-q)}-1}{1-q} - \frac{q}{2}\cdot 2^{n(1-q)} \rbra*{1-\frac{1}{n}}^2
        \coloneqq \tau_\ttN(n;q).
    \end{align*}
    \end{itemize}

    We then define the threshold function $t(n;q) \coloneqq \frac{1}{2}\rbra*{ \tau_\ttY(n;q) + \tau_\ttN(n;q)}$ and the gap function 
    \begin{align*}
    2 \cdot g(n;q) &\coloneqq \tau_\ttY(n;q) - \tau_\ttN(n;q) \\
    &= \frac{1}{1-q} \cdot \rbra*{\frac{1}{n} + 2^{-n}} - \frac{2^{n(1-q)}}{1-q} \rbra*{\frac{1}{n} + 2^{-n}} + \frac{q}{2}\cdot 2^{n(1-q)} \rbra*{1-\frac{1}{n}}^2\\
    &\geq \underbrace{\frac{1}{1-q} \cdot \rbra*{\frac{1}{n} + 2^{-n}}}_{G_1(n;q)} + \underbrace{2^{n(1-q)}}_{G_2(n;q)} \underbrace{\rbra*{ \frac{q}{2} \rbra*{1-\frac{1}{n}}^2 - \frac{2}{n(1-q)} }}_{G_3(n;q)}.
    \end{align*}
    Here, the inequality in the third line follows from the facts that $2^{n(1-q)}/(1-q) > 0$ for all $q\in(0,1)$ and $2^{-n} \leq 1/n$ for all $n\geq 1$. 

    Since $G_1(n;q)$ is non-negative and monotonically decreasing for $n\geq 1$, this term does not contribute to the lower bound. Moreover, as $G_2(n;q) \geq 0$ for all $n\geq 1$ and $q\in(0,1)$, it suffices to show that $G_3(n;q) > 0$ for sufficiently large $n$. 

    Since $n^2 G_3(n;q)$ can be written as a quadratic function of $n$,
    \[ n^2 G_3(n;q) = \frac{q}{2} \cdot n^2 - \rbra*{q+\frac{2}{1-q}} \cdot n + \frac{q}{2}, \]
    whose sign coincides with that of $G_3(n;q)$, it is evident that $n^2 G_3(n;q) > 0$ whenever 
    \[n \geq 
     n_0 \coloneqq 1 + \frac{2}{q(1-q)} \rbra*{1+\sqrt{1+q-q^2}}.\]
     
    Noting that $\sqrt{1+q-q^2} \in (1,\sqrt{5}/2]$ for $q\in(0,1)$, we may safely choose $n \geq n_\star \coloneqq \frac{5}{q(1-q)} > n_0$ for all $q\in(0,1)$, so that
    \[ G_3(n;q) \geq G_3(n_\star,q) = \frac{q}{2} \rbra*{1-\frac{q(1-q)}{5}}^2 - \frac{2}{5} q \geq \frac{q}{2} \rbra*{1-\frac{1}{20}}^2 - \frac{2}{5} q = \frac{41}{800}\cdot q > \frac{q}{20}.\]

    Hence, the gap function satisfies the desired lower bound: 
    \[ g(n;q) \geq \frac{1}{2} G_2(n;q) G_3(n;q) = \frac{1}{2} \cdot 2^{5/q} \cdot \frac{q}{20} = \frac{2^{5/q} q}{40}. \qedhere\]
\end{proof}

As a consequence, \Cref{lemma:reduction-MaxMixedQSD-TsallisQEA} implies the following \NIQSZK{}-hardness result: 

\TsallisQEAhardness*

\begin{proof}
    Using \Cref{QSCMM-is-NIQSZKhard}, we know that $\QSCMM[1/n, 1-1/n]$ is \NIQSZK{}-hard for $n \geq 3$. Combining this result with the reduction from \QSCMM{} to $\TsallisQEA$ for $q \in (0,1)$ and $n \geq \ceil[\big]{\frac{5}{q(1-q)}}$ (see \Cref{lemma:reduction-MaxMixedQSD-TsallisQEA}) and the specific choice of $t(n;q)$ in the reduction, we conclude that the resulting gap satisfies $g(n;q) \geq 2^{5/q}q/40$, completing the proof. 
\end{proof}

Directly analogous to \Cref{thm:TsallisQEA-NIQSZKhard}, the reduction in \Cref{lemma:reduction-MaxMixedQSD-TsallisQEA-half} implies the following: 

\begin{theorem}[$\TsallisQEAnoq_{1/2}$ is \NIQSZK{}-hard]
    \label{thm:TsallisQEA-NIQSZKhard-half}
    For all integers $n \geq 4$, it holds that\emph{:}
    \[ \forall g(n) \in[1/\poly(n), 1/5],\quad \TsallisQEAnoq_{1/2}[t(n),g(n)] \text{ is } \NIQSZK{}\text{-hard}. \]
    The threshold parameter $t(n)$ is as defined in \Cref{eq:tn-reduction-MaxMixedQSD-TsallisQEA-half}. 
\end{theorem}

\addcontentsline{toc}{section}{References}

\bibliographystyle{alphaurl}
\bibliography{main}

\end{document}